\def\year{2020}\relax
\documentclass[letterpaper]{article} 
\usepackage{aaai20}  
\usepackage{times}  
\usepackage{helvet} 
\usepackage{courier}  
\usepackage[hyphens]{url}  
\usepackage{graphicx} 
\usepackage{graphicx}  
\usepackage{amsmath}
\DeclareMathOperator*{\argmax}{argmax} 
\usepackage{amssymb}
\usepackage{amsthm}
\usepackage{tikz}
\usepackage{subcaption}

\newcommand{\citet}[1]{\citeauthor{#1} \shortcite{#1}}

\title{Bi-level Actor-Critic for Multi-agent Coordination \thanks{Weinan Zhang and Jun Wang are the corresponding authors of this paper.}}

\author{
Haifeng Zhang\textsuperscript{\rm 1},
Weizhe Chen\textsuperscript{\rm 2},
Zeren Huang\textsuperscript{\rm 2},
Minne Li\textsuperscript{\rm 1}, 
Yaodong Yang\textsuperscript{\rm 3},
Weinan Zhang\textsuperscript{\rm 2},
Jun Wang\textsuperscript{\rm 1}
\\  
\textsuperscript{\rm 1}University College London\\ 
\textsuperscript{\rm 2}Shanghai Jiao Tong University\\ 
\textsuperscript{\rm 3}Huawei R\&D UK\\
haifeng.zhang@ucl.ac.uk, chenweizhe@sjtu.edu.cn, sjtu\_hzr@sjtu.edu.cn, minne.li@cs.ucl.ac.uk \\ yaodong.yang@huawei.com, wnzhang@sjtu.edu.cn, junwang@ucl.ac.uk
}

\newtheorem{theorem}{Theorem}

\newtheorem{lemma}{Lemma}
\newtheorem{definition}{Definition}
\newtheorem{assumption}{Assumption}

\begin{document}

\maketitle

\begin{abstract}
    
    Coordination is one of the essential problems in multi-agent systems. Typically multi-agent reinforcement learning (MARL) methods treat agents equally and the goal is to solve the Markov game to an arbitrary Nash equilibrium (NE) when multiple equilibra exist, thus lacking a solution for NE selection. In this paper, we treat agents \emph{unequally} and consider Stackelberg equilibrium as a potentially better convergence point than Nash equilibrium in terms of Pareto superiority, especially in cooperative environments. Under Markov games, we formally define the bi-level reinforcement learning problem in finding Stackelberg equilibrium. We propose a novel bi-level actor-critic learning method that allows agents to have different knowledge base (thus intelligent), while their actions still can be executed simultaneously and distributedly. The convergence proof is given, while the resulting learning algorithm is tested against the state of the arts. We found that the proposed bi-level actor-critic algorithm successfully converged to the Stackelberg equilibria in matrix games and find an asymmetric solution in a highway merge environment. 
   
\end{abstract}

\section{Introduction}
    In a multi-agent system, the effect of any agent's action on the environment also depends on the actions taken by other agents and coordination is needed to consistently break ties between equally good actions or strategies \cite{bu2008comprehensive}. This problem is essential especially in the circumstances where the agents are not able to communicate. In game theory, coordination game is defined as the game with multiple Nash equilibria. Various criteria for Nash equilibrium selection were proposed in the game theory literature such as salience \cite{vanderschraaf1995convention} and fairness \cite{rabin1993incorporating}, where the agents are assumed to know the game model before applying these criteria. For the environments where agents are not able to know the game model but can learn it through interactions with the environments, multi-agent reinforcement learning approaches were proposed to find a Nash equilibrium, including Nash Q-learning \cite{hu2003nash}, MADDPG \cite{lowe2017multi} and the Mean-Field Q-learning \cite{yang2018mean}. These model-free approaches train the agents centrally to converge to a Nash equilibrium and then execute the agents distributively. However, these approaches can not guarantee a particular converged Nash equilibrium, which leads to uncertainty and sub-optimality. 
    
    To tackle this problem, we reconsider the coordination problem from an asymmetric angle. Although the original game model is symmetric that agents make decisions simultaneously, we are still able to define a priority of decision making for the agents in the training phase and keep simultaneous decision making in the execution phase. In this asymmetric game model, the Stackelberg equilibrium (SE) \cite{von2010market} is naturally set up as the learning objective rather than the Nash equilibrium (NE). The SE optimizes the leader's policy given that the follower always plays the best-response policy. Despite its discrimination on the follower, we surprisingly find the SE is Pareto superior than the NE in a wide range of environments. For example, in the cooperative games, the SE is guaranteed to be Pareto optimal, whereas only one of the NEs achieves this point, as Table \ref{table:coordination-game} shows. In a non-cooperative case shown in Table \ref{table:se-vs-ne}, the SE is not included in the set of the NEs and is Pareto superior to any NE. In general, our empirical study shows the SE is likely to be Pareto superior to the average NE in games with high cooperation level. 
    
    For solving the SE, a wide variety of bi-level optimization methods were proposed \cite{dempe2018bilevel}. However, our problem setting differs from the traditional bi-level optimization problem in two aspects: 1) we consider a multi-state environment where the objective function is a summation of the sequential discounted rewards; 2) our game model is unknown and can only be learned through interactions. Actually, the traditional bi-level optimization problem can be regarded as a stateless model-based version of our problem. We formally define our problem as the bi-level reinforcement learning problem and propose a novel bi-level actor-critic algorithm to solve it. We train the actor of the follower and the critics of both agents centrally to find an SE and then execute the agents distributively. Our experiments in the small environments and a simulated highway merge environment demonstrate the efficiency of our algorithm, outperforming the state-of-the-art algorithms.

    \begin{table}[t]
        \begin{subtable}{.5\linewidth}
            \centering
            \small
            \begin{tabular}{llll}
                                   & \multicolumn{1}{c}{X}       & \multicolumn{1}{c}{Y}       & \multicolumn{1}{c}{Z}       \\ \cline{2-4} 
            \multicolumn{1}{l|}{A} & \multicolumn{1}{l|}{15, 15} & \multicolumn{1}{l|}{10, 10}   & \multicolumn{1}{l|}{0, 0}   \\ \cline{2-4} 
            \multicolumn{1}{l|}{B} & \multicolumn{1}{l|}{10, 10}   & \multicolumn{1}{l|}{10, 10} & \multicolumn{1}{l|}{0, 0}   \\ \cline{2-4} 
            \multicolumn{1}{l|}{C} & \multicolumn{1}{l|}{0, 0}   & \multicolumn{1}{l|}{0, 0} & \multicolumn{1}{l|}{30, 30}   \\ \cline{2-4} 
            \end{tabular}
            \caption{}
            \label{table:coordination-game}
        \end{subtable}%
        \begin{subtable}{.5\linewidth}
            \centering
            \small
            \begin{tabular}{llll}
                                   & \multicolumn{1}{c}{X}       & \multicolumn{1}{c}{Y}       & \multicolumn{1}{c}{Z}       \\ \cline{2-4} 
            \multicolumn{1}{l|}{A} & \multicolumn{1}{l|}{20, 15} & \multicolumn{1}{l|}{0, 0}   & \multicolumn{1}{l|}{0, 0}   \\ \cline{2-4} 
            \multicolumn{1}{l|}{B} & \multicolumn{1}{l|}{30, 0}   & \multicolumn{1}{l|}{10, 5} & \multicolumn{1}{l|}{0, 0}   \\ \cline{2-4} 
            \multicolumn{1}{l|}{C} & \multicolumn{1}{l|}{0, 0}   & \multicolumn{1}{l|}{0, 0} & \multicolumn{1}{l|}{5, 10}   \\ \cline{2-4} 
            \end{tabular}
            \caption{}
            \label{table:se-vs-ne}
        \end{subtable}
        \caption{Coordination games. (a) A cooperative game where A-X and C-Z are the NE. C-Z is also the SE and the Perato optimality point. (b) A non-cooperative game where A-X is the SE, B-Y and C-Z are the NE. The SE is Pareto superior to any NE in this game.}
    \end{table} 
    
\section{Preliminaries}

\subsection{Markov Game}
    
    In an $n$-player Markov game \cite{littman1994markov} (or stochastic game) $\langle S, A_i, P, R_i, \gamma \rangle$, $S$ denotes the state space, $A_i$ denotes agent $i$'s action space and $A$ denotes the joint action space, $P:S \times A \rightarrow \text{PD}(S)$ \footnote{In this paper, $\text{PD}(X)$ denotes the probability distribution space over discrete set $X$.} denotes the transition function, $R_i:S \times A_i \rightarrow \mathcal{R}$ denotes the reward function for agent $i$, and $\gamma$ denotes the discount factor. Agents take actions simultaneously in each state following their policies $\pi_i:S \rightarrow \text{PD}(A_i)$. The objective of agent $i$ is to maximize its discounted cumulative reward $\sum_t \gamma^t r_i^{t}$, where $r_i^{t}$ is the reward agent $i$ receives in time-step $t$. We also call Markov games as multi-agent reinforcement learning (MARL) problems.

\subsection{Related MARL Solutions}

    For a Markov game, we have Bellman equations that characterize the optimal state-values $V_i^*(s)$ and action-values $Q_i^*(s,a)$:
    \begin{align}
        \label{eq:bellman}
        Q_i^*(s, \vec{a}) = R(s, \vec{a}) + \gamma \sum_{s'} P(s, \vec{a}, s') V_i^*(s').
    \end{align}

    The Minimax-Q method \cite{littman1994markov} attempts to find the highest worst case values in zero-sum games whose state-values are computed as:
    \begin{align}
        V_1^*(s) = \max_{\pi_1 \in \Pi_1} \min_{a_2 \in A_2} Q^*_1(s, \pi_1, a_2) = - V_2^*(s),
    \end{align}
    where $Q^*_1(s, \pi_1, a_2) = \sum_{a_1 \in A_1} \pi_1(s, a_1) Q^*_1(s, a_1, a_2)$ and $\Pi_1$ denotes the policy space of agent 1. Our bi-level method generalizes the minimax method from zero-sum games to general-sum games.
    
    The Nash-Q method \cite{hu2003nash} attempts to find the Nash equilibrium whose state-values are computed as:
    \begin{align}
        V_i^*(s) = \textit{NASH}_i\big(Q_1^*(s), Q_2^*(s), ..., Q_n^*(s)\big) 
    \end{align}
    where $\textit{NASH}_i(\vec{x}_1, \vec{x}_2, ..., \vec{x}_n)$ denotes the $i$-th agent's payoff in a Nash equilibrium of the matrix game formed by $\vec{x}_1, \vec{x}_2, ..., \vec{x}_n$. The Nash-Q method also generalizes the minimax method to general-sum games, but in a different direction compared to our method. Our bi-level method attempts to find Stackelberg Equilibrium rather than Nash equilibrium. 
    

\subsection{Bi-level Optimization}
    
    In this paper, we assume the agents in a two-player Markov game are asymmetric that the following agent observes the actions of the leading agent, which results in solving a bi-level optimization problem for a Markov game. The original bi-level optimization problem is formulated as below:
    \begin{align}
        \label{eq:multi-level-opt}
        \min_{x_1} & ~~~ f_1(x_1, x_2) \nonumber\\
        \text{s.t. } & ~~~ g_1(x_1, x_2) \le 0 \\
        & ~~~ \min_{x_2} ~~~ f_2(x_1, x_2) \nonumber\\
        & ~~~~~ \text{s.t. } ~~~ g_2(x_1, x_2) \le 0 \nonumber
    \end{align}
    where $f_i, i=1,2$ are the objective functions and $g_i, i=1,2$ are the constraint functions in each level. 
    
    The bi-level optimization problem can be equivalently described as a Stackelberg game where the upper-level optimizer is the leader and the lower-level optimizer is the follower and the solution of the bi-level optimization problem is the Stackelberg equilibrium.
    
\section{Bi-level Reinforcement Learning}

\subsection{Problem Formulation}
    
    Connecting bi-level optimization with Markov game, $x_i$ in Eq.~(\ref{eq:multi-level-opt}) corresponds to agent $i$'s policy $\pi_i$, $f_i$ corresponds to agent $i$'s cumulative reward and $g_i$ corresponds to the constraint of action space. Assuming Agent 1 as the leader and Agent 2 as the follower, our problem is formulated as:
    \begin{align}
        \label{eq:problem}
        \max_{\pi_1} & ~~~ \mathbb{E}_{r_1^{1},r_1^{2}... \sim \pi_1, \pi_2} \sum_{t=1}^{\infty} \gamma^t r_1^{t} \nonumber \\
        \text{s.t.} & ~~~\pi_1 \in \Pi_1  \\
        & ~~~\max_{\pi_2} ~~~ \mathbb{E}_{r_2^{1},r_2^{2}... \sim \pi_1, \pi_2} \sum_{t=1}^{\infty} \gamma^t r_2^{t} \nonumber \\
        & ~~~~~~~ \text {s.t.} ~~~ \pi_2 \in \Pi_2. \nonumber
    \end{align}
    We call this problem \emph{bi-level reinforcement learning} (BiRL). BiRL can be viewed as a multi-state version of the Stackelberg game \cite{von2010market} and extends the standard bi-level optimization problem in two dimensions: 1) the objective is a summation of the discounted rewards in sequential states; 2) the form of the objective function is unknown and can only be learned through interactions with the environment in a model-free way. The standard bi-level optimization problem can be viewed as a stateless model-based version of our problem. 
    
\subsection{Stackelberg Equilibrium vs. Nash equilibrium}
    
    We formulate BiRL to tackle the coordination problem in MARL. In game theory, coordination game is defined as a game with multiple Nash equilibria and the coordination problem can be regarded as a Nash equilibrium selection problem. In this paper, we consider Stackelberg equilibria as a potentially better solution for coordination games. Fig.~\ref{fig:grid-game} is an example demonstrating the difference between NE and SE in a Markov game. We figure out two advantages of the SE over the NE. 
    
    The first advantage of SE is the certainty. Multiple NEs exist in most games while multiple SEs exist only in a class of games. \footnote{Multiple SEs only exist when given the policy of the leader, multiple policies of the follower achieve the maximal payoff, or given the best response of the follower, multiple policies of the leader achieve the maximal payoff.} Existing MARL methods mainly converge to an arbitrary NE, which leads to uncertainty. Since the SE is less multiple in most games, it is more clear and stable to be a learning objective. By setting the SE as the objective, we actually attempt to avoid the coordination problem (or the NE selection problem) rather than solving it.
    
    The second advantage of SE is the performance. The SE may achieve better payoff than the average NE in coordination environments in terms of Pareto superiority. An extreme example is the cooperative games. The SE always achieves the Pareto optimality point in a cooperative game while only the best NE achieves so, as we showed in Table \ref{table:se-vs-ne} and Fig.~\ref{fig:grid-game}. In other words, both the leader and the follower achieve higher payoffs in the SE than in the average NE. We intuitively believe that this result would still hold in games with less (but still high) cooperation levels. 
        
    For demonstrating our belief, we formally define the cooperation level of two-player Markov games as the correlation between the cumulative rewards of the agents:
    \begin{align}
        \textit{CL} = \frac{\sum_{\vec{\pi}} (V_1^{\vec{\pi}} - \bar{V}_1)(V_2^{\vec{\pi}} - \bar{V}_2)}{\sqrt{\sum_{\vec{\pi}}(V_1^{\vec{\pi}} - \bar{V}_1)^2 \sum_{\vec{\pi}}(V_2^{\vec{\pi}} - \bar{V}_2)^2}}
    \end{align}
    where $V_i^{\vec{\pi}}$ is short for $V_i^{\vec{\pi}}(s_0)$ denoting the average discounted cumulative reward for agent $i$ from the start state $s_0$ following the joint policy $\vec{\pi}$ and $\bar{V}_i=\frac{1}{|\vec{\pi}|}\sum_{\vec{\pi}}V_i^{\vec{\pi}}$. Under this definition, the cooperation levels of a cooperative game and a zero-sum game are 1 and -1 respectively. 
    
    We empirically study the relationship between the cooperation level of a game and the average payoff achieved by the agents in the average NE and the SE. The results in Fig.~\ref{fig:se-vs-ne} demonstrate that both the leader and the follower achieve higher payoff in the SE not only in fully cooperative games but also in the games with high cooperation level. We also find that the number of Nash equilibiria in a game is positively correlated with the cooperation level, which suggests that the coordination problem is more likely to occur in games with high cooperation level. Hence, we argue that the SE may in general be Pareto superior to the average NE in coordination problems, especially in highly cooperative games. 
        
    \begin{figure}[h]
        \centering
        \includegraphics[width=0.3\columnwidth]{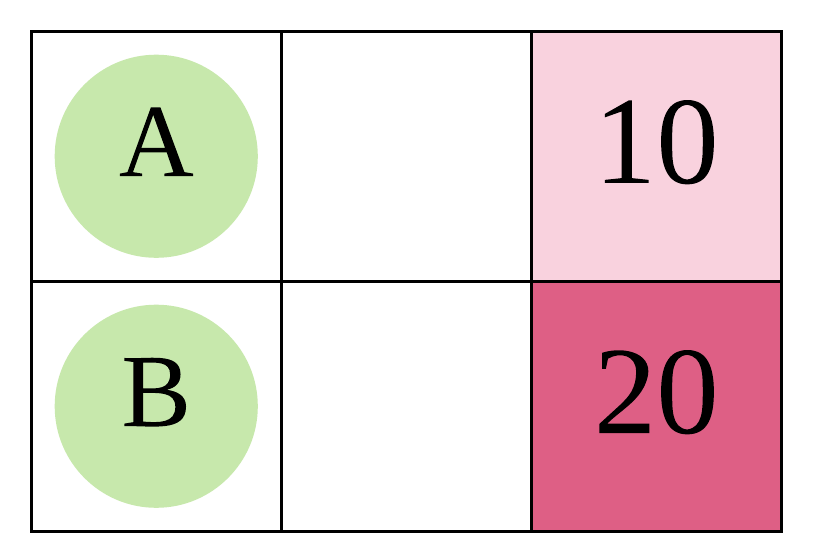}
        \caption{A cooperative game example for BiRL. Agent A and B move simultaneously in the grid and receive the common reward only when they are both in the 10 or 20 square. Joint policies lead both agents to 10 or 20 are Nash equilibria, but only joint policies lead both agents to 20 are Stackelberg equilibria and are the solutions for BiRL.}
        \label{fig:grid-game}
    \end{figure}

    \begin{figure}[h]
        \centering
        \includegraphics[width=1.0\columnwidth]{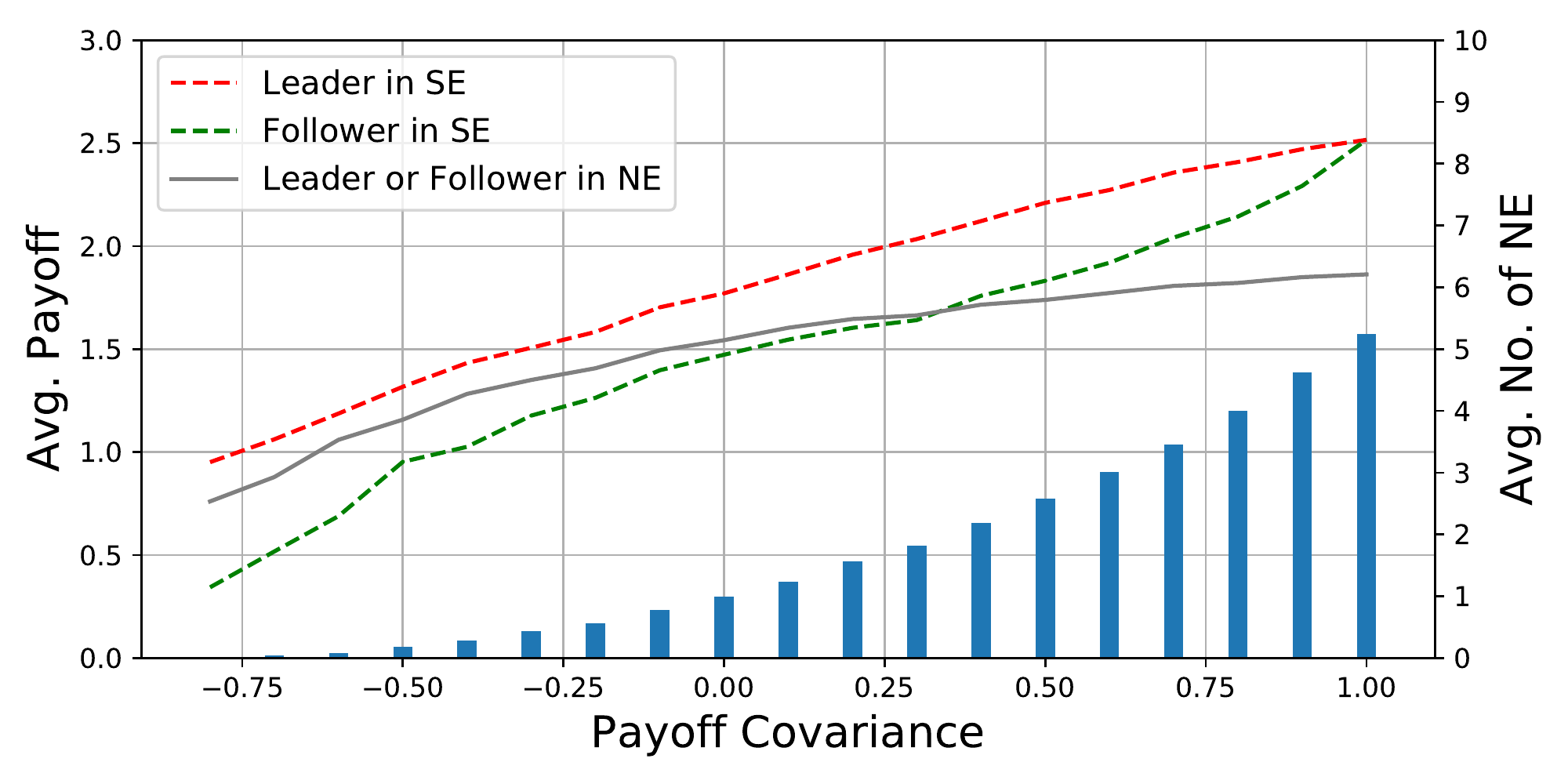}
        \caption{SE vs. NE. We sample the payoffs for a $10 \times 10$ matrix game using the multivariate normal distribution with 0 for mean, 1 for variance and various parameters for covariance, which represents the cooperation level of the generated matrix game. When the covariance equals to 1, the game is fully cooperative. The result for each covariance parameter is averaged from $2000$ independent trials. Similar results are found in matrix games with different sizes ranging from $5 \times 5$ to $100 \times 100$, which is provided in the supplementary. }
        \label{fig:se-vs-ne}
    \end{figure}

\subsection{Bi-level Tabular Q-learning}

    
    
    
    Similar to the minimax-Q and Nash-Q, we can define the bi-level Bellman equation by specifying the calculation method for the optimal state-values in Eq.~(\ref{eq:bellman}):
    \begin{align}
        \label{eq:bellman-stackelberg}
        V_i^*(s) = \textit{Stackelberg}_i(Q_1^*(s), Q_2^*(s)),
    \end{align}
    where $\textit{Stackelberg}_i(\vec{x}_1, \vec{x}_2)$ denotes the $i$-th agent's payoff in the Stackerberg Equilibrium of the matrix game formed by $\vec{x}_1, \vec{x}_2$.
    

    Based on the bi-level Bellman equation, we are able to update the Q-values iteratively by Eq.~(\ref{eq:bellman}) and (\ref{eq:bellman-stackelberg}). Formally, We have the update rules for $Q_1$ and $Q_2$ tables given a transaction $\langle s, a_1, a_2, s', r_1, r_2 \rangle$ with learning rate $\alpha_i$:
    \begin{align}
        &a_1' \leftarrow \argmax_{a_1} Q_1(s', a_1, \argmax_{a_2} Q_2(s', a_1, a_2)), \label{eq:select_a1} \\ 
        &a_2' \leftarrow \argmax_{a_2} Q_2(s', a_1', a_2), \label{eq:select_a2} \\ 
        &Q_1(s, a_1, a_2) \leftarrow (1-\alpha_1)Q_1(s, a_1, a_2) \nonumber\\ 
        &~~~~~~~~~~~~~~~~~~~~~~~~~~+ \alpha_1 (r_1 + \gamma Q_1(s', a_1', a_2')), \\
        &Q_2(s, a_1, a_2) \leftarrow (1-\alpha_2)Q_2(s, a_1, a_2)  \nonumber\\ 
        &~~~~~~~~~~~~~~~~~~~~~~~~~~+ \alpha_2 (r_2 + \gamma Q_2(s', a_1', a_2')).
    \end{align}
    This tabular method was also studied in \cite{littman2001leading} and \cite{kononen2004asymmetric}. However, these works mainly focused on solving asymmetric problems while our motivation is to solve symmetric coordination problems using an asymmetric method. 
    
\subsection{Bi-level Actor-Critic}

    \begin{figure}
        \centering
        \includegraphics[width=0.9\columnwidth]{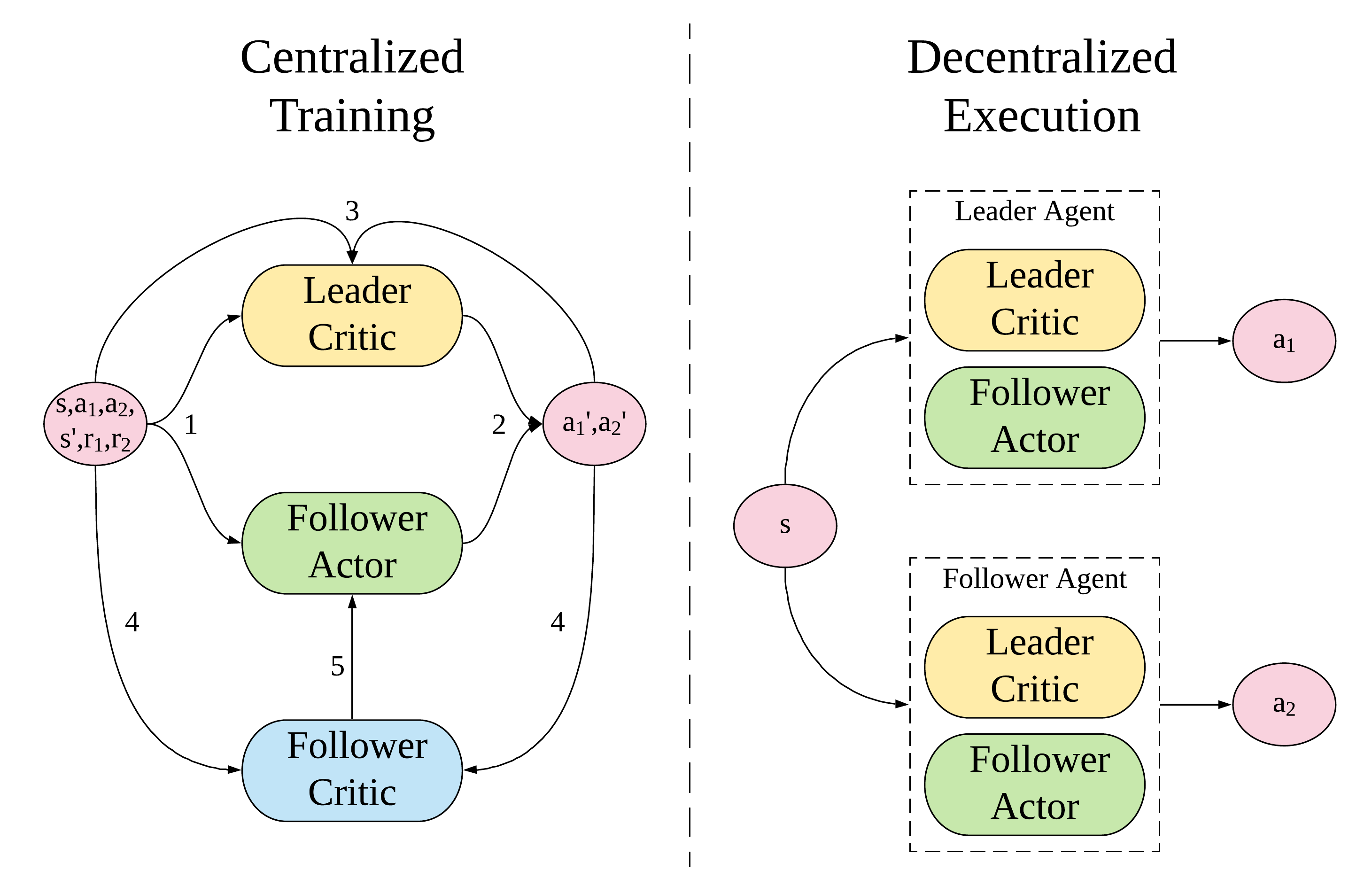}
        \caption{Structure of bi-level actor-critic. In the training phase, the joint action in the next state is computed firstly, then the three models are updated accordingly. In the execution phase, both agents have the leader critic and the follower actor models and perform Stackelberg equilibria independently.}
        \label{fig:bi-ac}
    \end{figure}
    
    In Eq.~(\ref{eq:select_a1}), we need to enumerate the actions in both levels to select action $a_1'$, which leads to $|A_1| \cdot |A_2|$ visits to the $Q_2$ table. When $Q_2$ is modeled by an approximation function, i.e. a neural network, the calculations of Eq.~(\ref{eq:select_a1}) could be time-consuming. Furthermore, if we extend the bi-level Q-learning methods to multi-level, the computation complexity of $a_1'$ would increase in exponential w.r.t. the number of level. 
    
    For solving this problem, we propose the bi-level actor-critic (Bi-AC) method which introduces an actor for the follower while keeping the leader as a Q-learner. Formally, let $\pi_2(s, a_1;\phi_2) \in \text{PD}(A_2)$ denote the policy model (or actor) of agent 2, which takes agent 1's action as its input in addition to the current state. We also model the two critics using approximation functions for both agents. We have the following update rules given a transaction $\langle s, a_1, a_2, s', r_1, r_2 \rangle$ with learning rate $\alpha_i, \beta$:
    \begin{align}
        &a_1' \leftarrow \argmax_{a_1} Q_1(s', a_1, \pi_2(s', a_1; \phi_2); \theta_1), \\
        &a_2' \leftarrow \pi_2(s', a_1';\phi_2), \label{eq:select_a2_large} \\
        &\delta_i \leftarrow r_i + \gamma Q_i(s', \vec{a}';\theta_i) - Q_i(s, \vec{a};\theta_i),i=1,2, \label{eq:fa-delta} \\
        & \theta_i \leftarrow \theta_1 + \alpha_i \delta_i \nabla_{\theta_i} Q_i(s, \vec{a}; \theta_i), i=1,2, \label{eq:fa-theta} \\
        &\phi_2 \leftarrow \phi_2 + \beta \nabla_{\phi_2} \text{log} \pi_2(s, \vec{a}; \phi_2)Q_2(s, \vec{a}; \theta_2).
    \end{align}
    where $\pi_2'(s, a_1; \phi_2)$ is modeled by a Gumbel-Softmax estimator \cite{jang2016categorical} which computes $a_2'$ directly. 

    For the environments with continuous action space, we model agent 2's policy using a deterministic model $\mu_2(s, a_1; \phi_2) \in A_2$ which is updated by the deterministic policy gradient method \cite{silver2014deterministic}. The Q-network of agent 1 can be updated by the soft Q-learning \cite{haarnoja2017reinforcement} method.
    
    Bi-AC is a centralized-training-decentralized-execution method as Fig.~ \ref{fig:bi-ac} shows. The three models are trained together given off-policy episodes. In execution, the trained leader critic model and follower actor model are both allocated to and executed by the leader and the follower. In such way, the two agents are able to achieve the Stackelberg equilibrium distributively. 
    
    In the partially observable environments, we train two additional actors $\pi_i'(o_i)$ as the approximators for each agent, where the corresponding critics are the trained leader and follower critics. The approximators allow the two agents play joint actions forming Stackerberg equilibria based on their own observations.

    

    Bi-AC can be naturally extended to $n$-level actor-critic. In the case of continuous action space, we define deterministic policy models:
    \begin{align}
        \mu_i(s, a_1, a_2, ... a_{i-1}; \theta_i), i=1..n.
    \end{align}
    We also model the Q-functions for each agent as:
    \begin{align}
        Q_i(s, a_1, a_2, ... a_n; \phi_i), i=1..n.
    \end{align}
    In each training step, the actions in the next step are determined one by one from the upper-level agent to the lower-level agent and the models are updated accordingly: 
    \begin{align}
        &a_i' \leftarrow \mu_i(s', a_{1..i-1}'; \theta_i), \\
        &\delta_i \leftarrow r_i  + \gamma Q_i(s', \vec{a}';\phi_i) - Q_i(s, \vec{a};\phi_i), \\
        &\theta_i \leftarrow \theta_i + \alpha \delta_i \nabla_{\theta_i} Q_i(s, \vec{a}; \theta_i), \\
        &\phi_i \leftarrow \phi_i + \alpha \nabla_{\phi_i}\mu_i(s, a_{1..i-1}; \phi_i) \cdot \nabla_{a_i}Q_i(s, \vec{a} ;\theta_i).
    \end{align}
    
    In practical problems, the models can be modified slightly to contain multiple agents in each level, where agents in the same level take actions simultaneously and the lower-level agents observe the actions of the upper-level agents. 
    
\subsection{Convergence and Limitation}
    
    Bi-AC will converge to the Stackelberg equilibrium under the following assumptions:
    \begin{enumerate}
        \item Every stage game $(Q^t_1(s), Q^t_2(s))$ for all t and s, has a global optimal point , and agents' payoffs in this point are selected by the actor function to update the critic functions with probability $1$.
        \item The critic learning rates $\alpha_{t}$ for the $t$-th transaction satisfies $\sum_{t=0}^{\infty}\alpha^t(s, \vec{a})=\infty$, $\sum_{t=0}^{\infty}[\alpha^{t}(s, \vec{a})]^2 < \infty$ hold uniformly and with probability $1$.
        \item Every stage game $(Q_1^t(s), Q_2^{t}(s))$ for all $t$ and $s$ has a global optimal point, and the agents' payoffs in this point are selected by Bi-AC to update the critic functions with probability $1$.
    \end{enumerate}
    Given these assumptions, we can use the method of Cauchy sequence convergence to prove the convergence of our algorithm. The detailed proof and discussion are provided in the supplementary. 
    
    The assumption $3$ above is a strong assumption which may not be met even in the very start of the training process. However, as shown later in the experiment part, our algorithm does converge in many cases. 
    
    Our algorithm is based on the bi-level Bellman equation, which is a necessary but not sufficient condition of the solution of BiRL. Therefore, there may exist convergent points other than the solution of BiRL. We provide an example in the supplementary.

\section{Related Work}

    In MARL, various approaches were proposed to tackle the coordination problem \cite{bu2008comprehensive}, especially for the cooperative environments. A general approach is applying the social convention which breaks ties by ordering of agents and actions \cite{boutilier1996planning}. Our method is compatible with social convention in the sense that we find the SE as the common knowledge of the agents about the game, based on which they can form social conventions. For cooperative games, the optimistic exploration were proposed \cite{claus1998dynamics} for reaching optimal equilibrium. \citet{lauer2000algorithm} used maximal estimation to update Q-value which ensures convergence to the optimal equilibrium given the reward function is deterministic. For the case of stochastic reward function, FMQ \cite{kapetanakis2002improving}, SOoN \cite{matignon2009coordination} and LMRL2 \cite{wei2016lenient} were proposed. These works share the idea of optimistic expectation on the cooperative opponent, which could not be extended to general-sum games. 

    Communication is an essential method to facilitate coordination. CommNet \cite{sukhbaatar2016learning} used a centralized network architecture to exchange information between agents. BicNet \cite{peng2017multiagent} proposed the bidirectional RNNs to exchange information between agents in an actor-critic setting. MADDPG \cite{lowe2017multi} proposed the centralized-training-decentralized-execution scheme which is also adopted by our method. Other works in this area include DIAL \cite{foerster2016learning}, COMA \cite{foerster2018counterfactual} and MD-MADDPG \cite{pesce2019improving}. 
    In addition, understanding other agents in multi-agent environments is of vital importance. Opponent modeling methods are helpful for coordination in many circumstances. ROMMEO \cite{tian2019regularized} applied maximum entropy to model the opponent. \citet{wen2019probabilistic} introduced the idea of recursive reasoning between two agents. Opponent modeling methods adopt the decentralized training scheme while our training is centralized. 

    Our work tackle the coordination problem from a asymmetric angle. BiRL is an extension of the bi-level optimization \cite{dempe2018bilevel} or the Stackelberg game \cite{von2010market}. For solving the original bi-level problem, stateless model-based evolutionary algorithms were proposed, such as BLEAQ \cite{sinha2014improved}. Extensively, a stateless model-free leader-follower problem was studied \cite{Zhang2012BilevelDS} where the objective functions are not visible. In the other dimension, BiMPC \cite{mintz2018control} studied the multi-state model-based Stackelberg game under the linear-quadratic assumption. In our paper, we formulate the multi-state model-free problem of BiRL, which extends the original bi-level problem in both dimensions.

    Our Bi-AC method contains critics which are iteratively updated by the Bellman equation. There are a series of MARL methods adopting the similar update scheme. Minimax-Q \cite{littman1994markov} solved the two-player zero-sum Markov games. Afterwards, fiend-and-foe learning \cite{littman2001friend}, Nash-Q \cite{hu2003nash}, CE-Q \cite{greenwald2003correlated}, Coco-Q \cite{sodomka2013coco} and AQL \cite{kononen2004asymmetric} were proposed successively. Among them, AQL updates the Q-value by solving the Stackelberg Equilibrium in each iteration, which can be regarded as the value-based version of Bi-AC. Compared to AQL, Bi-AC is able to work in multi-level or continuous action space environments. Another difference between AQL and our work is the motivation that we propose the Stackelberg Equilibrium as a potentially better solution for solving symmetric coordination problems while AQL focused on asymmetric problems. Other works applying Stackelberg Equilibrium to solve asymmetric problems include Bully \cite{littman2001leading} and DeDOL \cite{wang2019deep}.
    
\section{Experiment}

    We performed the experiments in three coordination environments comparing Bi-AC with the state-of-the-art MARL algorithms. \footnote{Our experiments are repeatable and the source code is provided in https://github.com/laonahongchen/Bilevel-Optimization-in-Coordination-Game.} 
    
\subsection{Algorithms}
    
    We compared Bi-AC with I-DQN, MADDPG, TD3, Hysteretic-DQN and Lenient-DQN, where I-DQN, Hysteretic-DQN and Lenient-DQN are decentralized training algorithms and MADDPG and TD3 are centralized training algorithms. In each experiment, we fully explored the actions for 1000 steps in the beginning and then applied the decaying $\epsilon$-greedy method for explorations. A three layer fully connected neural network with ReLU \cite{nair2010rectified} activation function was applied for the models in each algorithm to approximate the actor and critic functions. We trained the critic models with the replay buffer and the target network introduced in DQN \cite{mnih2015human}. We used the Gumbel-Softmax estimator \cite{jang2016categorical} in the actor function when applicable.

\subsubsection{Bi-AC}

    We realized the Bi-AC algorithm where the two critics were modeled by DQN. 
        

\subsubsection{I-DQN}    

    We tested independent DQN which regards other agents as a part of the environments. I-DQN is not guaranteed to converge and if it converges it will converge to a Nash equilibrium. 

\subsubsection{MADDPG}

    We tested MADDPG \cite{lowe2017multi} as the baseline of the centralized-training-decentralized-execution algorithm. 
    
\subsubsection{TD3}

    Twin Delayed Deep Deterministic policy gradient (TD3) \cite{fujimoto2018addressing} enhances MADDPG by addressing function approximation error, which is the state-of-the-art MARL algorithm. 

\subsubsection{Hysteretic-DQN}

    Hysteretic-DQN \cite{DBLP:conf/icml/OmidshafieiPAHV17} is based on hysteretic Q-learning \cite{DBLP:conf/iros/MatignonLF07} where a smaller learning rate is used when an update would reduce a Q-value, leading to an optimistic update function which puts more weights on positive experiences.
    
\subsubsection{Lenient-DQN}

    Lenient-DQN \cite{DBLP:conf/atal/PalmerTBS18} combines DQNs with leniency, which increases the likelihood of convergence towards the global optimal solution within fully cooperative environments that require implicit coordination strategies.

    \begin{figure}[t]
        \centering
        \includegraphics[width=0.9\columnwidth]{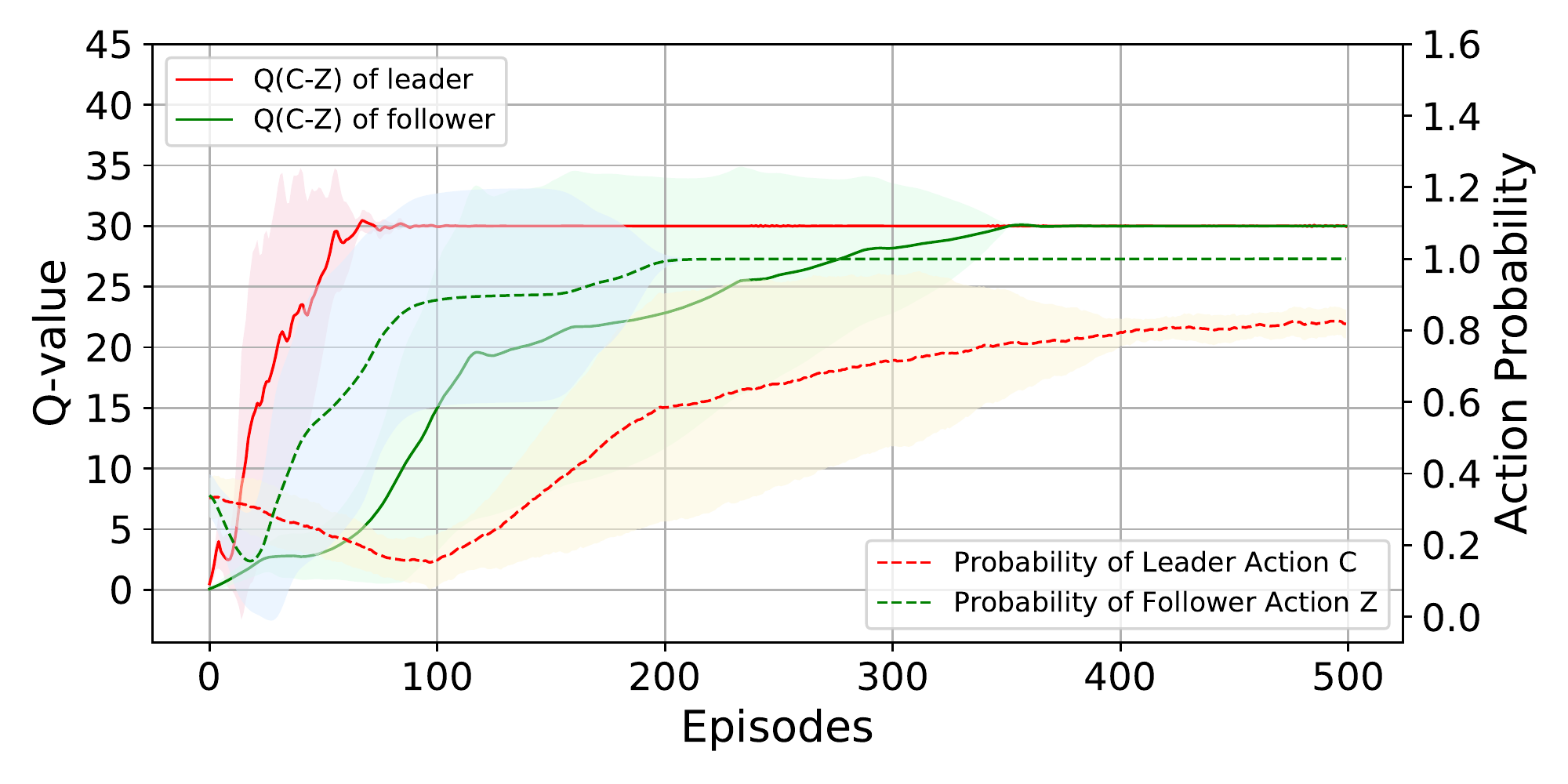}
        \caption{Convergence of Bi-AC in Escape game}
        \label{fig:matrix_1}
    \end{figure}{}

    \begin{figure}[t!]
    \begin{subfigure}{.5\textwidth}
      \centering
      \includegraphics[width=.8\linewidth]{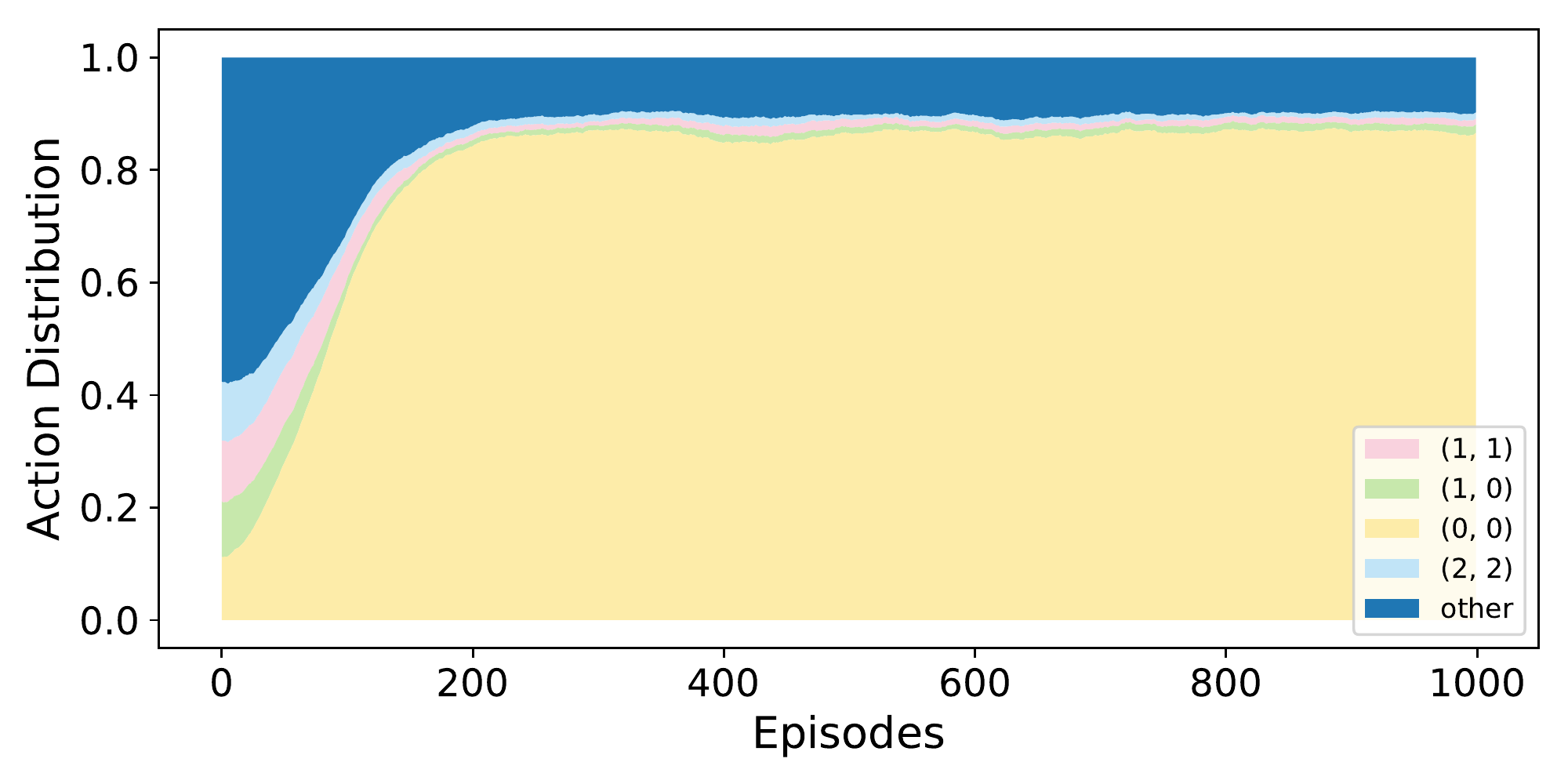}  
      \caption{Bi-AC}
      \label{fig:maxtrix2-biac}
    \end{subfigure}

    \begin{subfigure}{.5\textwidth}
        \centering
      \includegraphics[width=.8\linewidth]{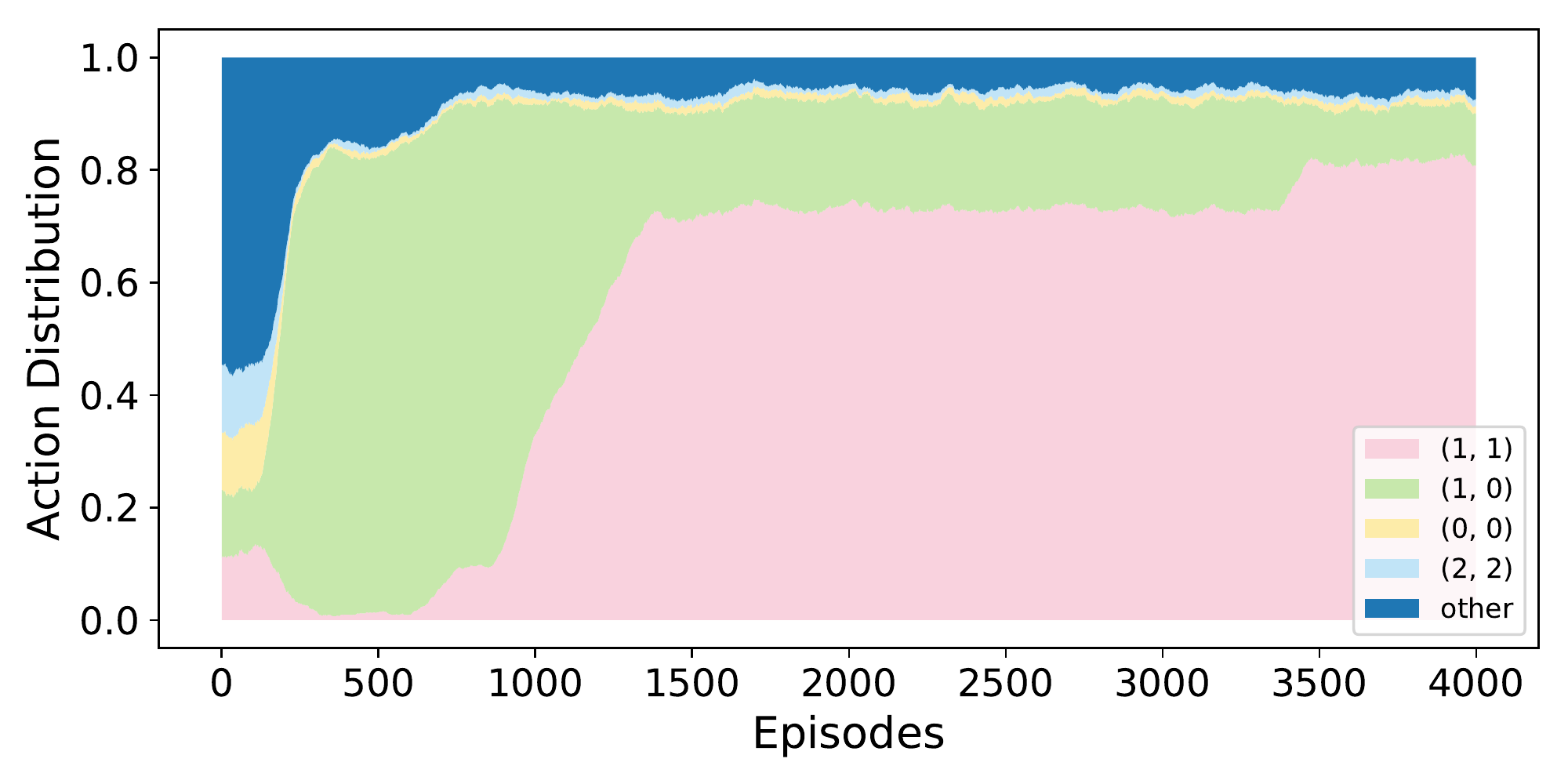}  
      \caption{MADDPG}
      \label{fig:matrix2-ddpg}
    \end{subfigure}

    \caption{Distribution of entry selected by different algorithms while training. Each color represent the possibility in the recent 100 steps to choose the corresponding entry.}
    \label{fig:matrix_game2}
    \end{figure}
    
    \begin{table}[t]
    \centering
    \small
    \begin{tabular}{llllll}
    \hline
    
    \multicolumn{1}{|l|}{}   & \multicolumn{1}{|l|}{Leader}   & \multicolumn{1}{|l|}{Follower} & \multicolumn{1}{|l|}{Optimality} \\ \cline{1-4} 
    \multicolumn{1}{|l|}{Bi-AC}   & \multicolumn{1}{l|}{28.5}   & \multicolumn{1}{l|}{ 28.5}   & \multicolumn{1}{l|}{90\%}  \\ \cline{1-4} 
    \multicolumn{1}{|l|}{I-DQN}   & \multicolumn{1}{l|}{20.55}   & \multicolumn{1}{l|}{ 20.55}   & \multicolumn{1}{l|}{37\%}  \\ \cline{1-4} 
    \multicolumn{1}{|l|}{MADDPG}   & \multicolumn{1}{l|}{23.4}   & \multicolumn{1}{l|}{23.4}   & \multicolumn{1}{l|}{56\%}   \\ \cline{1-4} 
    \multicolumn{1}{|l|}{TD3}   & \multicolumn{1}{l|}{15.3}   & \multicolumn{1}{l|}{15.3}   & \multicolumn{1}{l|}{2\%}   \\ \cline{1-4} 
    \multicolumn{1}{|l|}{Hysteretic-DQN}   & \multicolumn{1}{l|}{13.25}   & \multicolumn{1}{l|}{13.25}   & \multicolumn{1}{l|}{0\%}   \\ \cline{1-4} 
    \multicolumn{1}{|l|}{Lenient-DQN}   & \multicolumn{1}{l|}{13.4}   & \multicolumn{1}{l|}{13.4}   & \multicolumn{1}{l|}{0\%}   \\ \cline{1-4} 

    \end{tabular}
    \caption{Result of Escape game. The first two column means the average reward while the third column means the percentage of converging to C-Z, the global optimal point.}
    \label{table:Matrix_game_1}
    \end{table}

    \begin{table}[t]
    \centering
    \small
    \begin{tabular}{llllll}
    \hline
    \multicolumn{1}{|l|}{}   & \multicolumn{1}{|l|}{Leader}   & \multicolumn{1}{|l|}{Follower} & \multicolumn{1}{|l|}{Optimality} \\ \cline{1-4} 
    \multicolumn{1}{|l|}{Bi-AC}   & \multicolumn{1}{l|}{20}   & \multicolumn{1}{l|}{15}   & \multicolumn{1}{l|}{100\%}  \\ \cline{1-4} 
    \multicolumn{1}{|l|}{I-DQN}   & \multicolumn{1}{l|}{9.5}   & \multicolumn{1}{l|}{ 5.5}   & \multicolumn{1}{l|}{0\%}  \\ \cline{1-4} 
    \multicolumn{1}{|l|}{MADDPG}   & \multicolumn{1}{l|}{12}   & \multicolumn{1}{l|}{4.5}   & \multicolumn{1}{l|}{0\%}   \\ \cline{1-4} 
    \multicolumn{1}{|l|}{TD3}   & \multicolumn{1}{l|}{11.5}   & \multicolumn{1}{l|}{5}   & \multicolumn{1}{l|}{0\%}   \\ \cline{1-4} 
    \multicolumn{1}{|l|}{Hysteretic-DQN}   & \multicolumn{1}{l|}{9.1}   & \multicolumn{1}{l|}{5.9}   & \multicolumn{1}{l|}{0\%}   \\ \cline{1-4}
    \multicolumn{1}{|l|}{Lenient-DQN}   & \multicolumn{1}{l|}{8.8}   & \multicolumn{1}{l|}{6.2}   & \multicolumn{1}{l|}{0\%}   \\ \cline{1-4} 
 
    \end{tabular}
    \caption{Result of Maintain game.The first two column means the average reward while the third column means the percentage of converging to A-X.}
    \label{table:Matrix_game_2}
    \end{table}
    
    \begin{table}[t]
    \centering
    \small
    \begin{tabular}{lllll}
    \hline
    \multicolumn{1}{|l|}{}   & \multicolumn{1}{|l|}{Leader}   & \multicolumn{1}{|l|}{Follower}  \\ \cline{1-3} 
    \multicolumn{1}{|l|}{Bi-AC}   & \multicolumn{1}{l|}{72.2\%}   & \multicolumn{1}{l|}{25.2\%}     \\ \cline{1-3} 
    \multicolumn{1}{|l|}{I-DQN}   & \multicolumn{1}{l|}{58.2\%}   & \multicolumn{1}{l|}{41.3\%}    \\ \cline{1-3} 
    \multicolumn{1}{|l|}{MADDPG}   & \multicolumn{1}{l|}{37.5\%}   & \multicolumn{1}{l|}{60.2\%}     \\ \cline{1-3} 
    \multicolumn{1}{|l|}{TD3}   & \multicolumn{1}{l|}{38.6\%}   & \multicolumn{1}{l|}{55.1\%}    \\ \cline{1-3} 
    \end{tabular}
    \caption{Result of Traffic Merge. The Leader column shows the rate that the car goes first is from the main lane. The sum of the two columns is not 100\% due to a small rate of crash.}
    \label{table:Traffic}
    \end{table}

\subsection{Matrix Game}

    We tested two matrix games named Escape and Maintain whose payoff tables are shown in Table \ref{table:coordination-game} and Table \ref{table:se-vs-ne} respectively. 
    
    We designed the Escape game in Table \ref{table:coordination-game} to show that our algorithm has the capability to converge to the SE which is better than the average NE in a cooperative game. We expected Bi-AC converged to the C-Z point while other algorithms converged to either the A-X or the C-Z point. Note that the higher-left $2 \times 2$ part of the matrix will lead to the sub-optimal A-X point for an NE learner, unless the joint action C-Z is explored enough for an escape. 
    
    We ran each algorithm for 100 times and the results are provided in Table \ref{table:Matrix_game_1}. We see that Bi-AC achieved higher rewards for both agents than all the baseline algorithms. Also, Bi-AC converged to the global optimal point C-Z in $90\%$ trials leading the baseline algorithms by a large margin. Note that Bi-AC did not converge to the optimal point with $100\%$ probability because of the usage of neural network function approximations. \footnote{We performed the tabular version of Bi-AC to confirm this.} We also found MADDPG outperformed Hysteretic-DQN, Lenient-DQN and TD3 in this game. The reason may be that these three algorithms converges with lower bias which could also make it more difficult to converge to the isolated optimal point C-Z. The convergence curve of Bi-AC is provided in Fig. ~\ref{fig:matrix_1}. From the very start of training, the Q-value was trained to the correct estimated value $30$. And once the Q-value of the follower was trained to make C-Z higher than C-X and C-Y, the follower started to learn the best response toward action C of the leader. Then after the follower chose Z as the response with high probability, the probability for leader to choose action C started to increase, from about the 100th episode. The reason that the Q-value of the leader for C-Z converged before the follower's Q-value is that the learning rate for Q-function of the leader is higher than the follower. This asymmetric parameter setting is because the asymmetric architecture in our algorithm.
    We also find that the probability of action C of leader agent did not converge to $1.0$. It was because there was still exploration. 
    
    We designed the Maintain game in Table \ref{table:se-vs-ne} to show that Bi-AC algorithm is able to achieve the SE which is Pareto superior to all the NEs. As discussed, the A-X point is the SE and Pareto optimality point but not an NE. We expected our algorithm to converge to the A-X point while other NE learners converge to B-Y or C-Z. Particularly, NE learners are not able to maintain in A-X because B-X is a better point for the row player. The result is provided in table \ref{table:Matrix_game_2}. We can see that Bi-AC achieved higher rewards for both agents and converged to the A-X point in all the trials while the baselines never converged to the optimality. Fig. ~\ref{fig:matrix_game2} shows the learning processes of the joint policies for Bi-AC and MADDPG. Bi-AC converged to the A-X point smoothly while MADDPG converged to the B-Y point as expected. 
    
    \begin{figure}[t]
    \centering
    \includegraphics[width=0.45\textwidth]{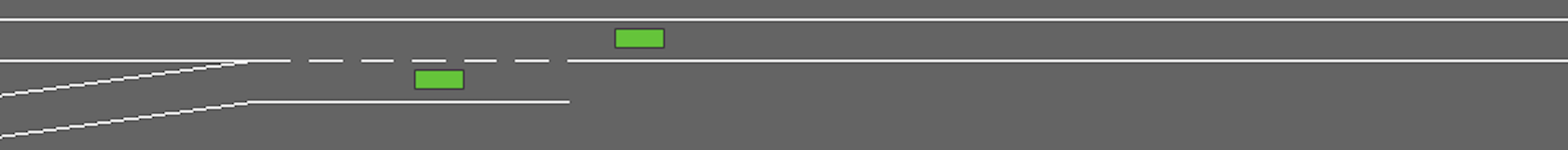}
    \caption{An illustration of the Traffic Merge environment.}
    \label{fig:highway_env}
    \end{figure}

\subsection{Highway Merge}
    
    \begin{figure}[t]
    \centering
    \includegraphics[width=0.45\textwidth]{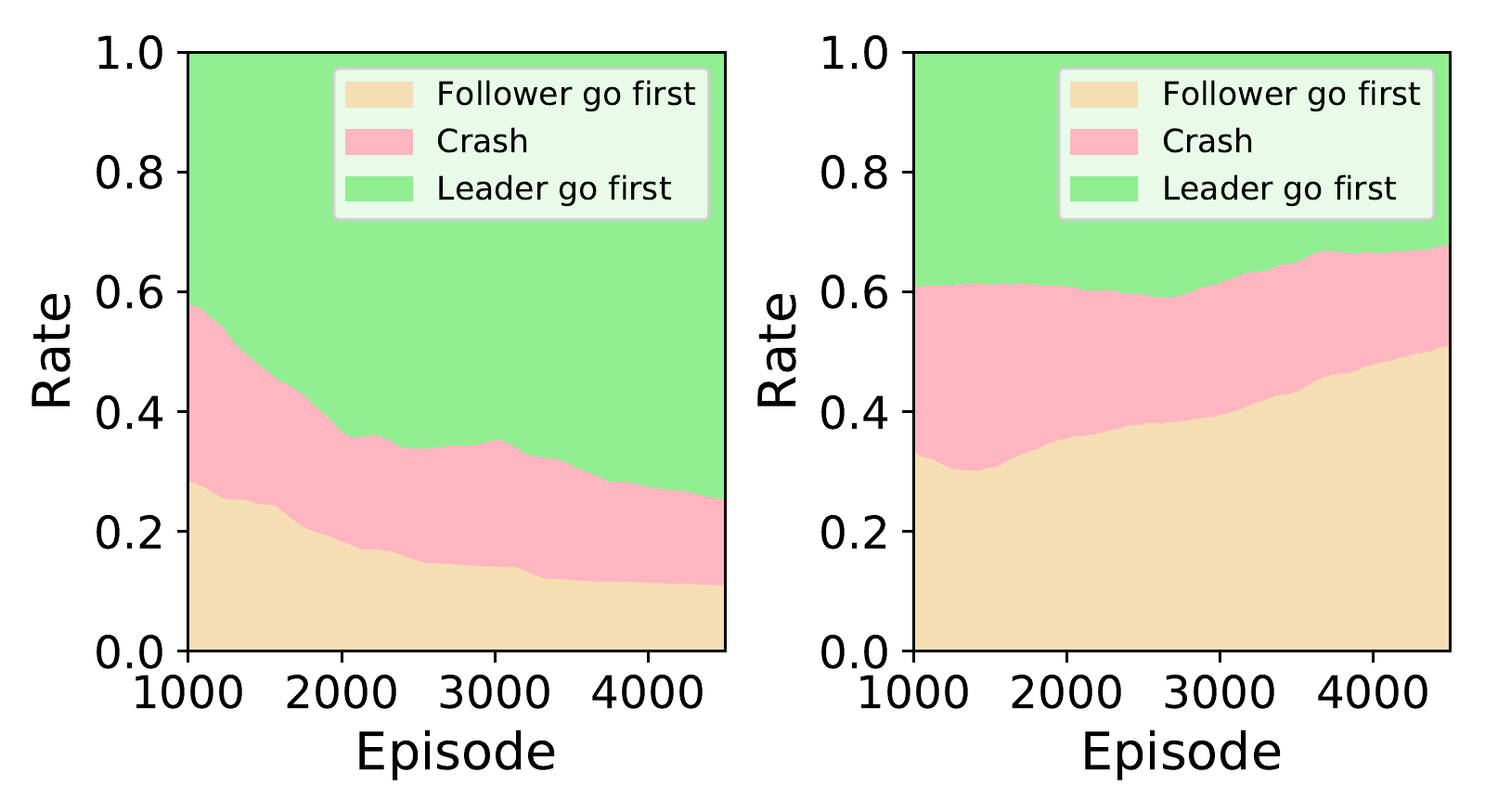}
    \caption{The training curve of highway merge problem. The left and right figures are the curve of Bi-AC and MADDPG respectively. The figures of other baselines are similar with MADDPG.}
    \label{fig:highway_merge_rate}
    \end{figure}

    We also designed a driving environment, where one car drives on the main lane and another car from the auxiliary lane wants to merge into the main lane. We used a slightly modified version of the Highway environment \cite{highway-env}, in which an agent can observe the kinematics of the nearby agent including its position and velocity, and has a discrete action space of $5$ actions including: LANE LEFT, LANE RIGHT, FASTER, SLOWER and IDLE, where IDLE means neither changing the current speed nor changing the LANE it is running on. The agents are rewarded for $50$ if it passes first, rewarded for $10$ if it passes second, and rewarded for $-10$ if two cars bump together as a penalty. An example overview of the environment is given in Fig.~\ref{fig:highway_env}. Note that if the cars from auxiliary lane do not choose LANE LEFT before his own lane disappears, the environment will automatically drive the cars to perform LANE LEFT right before they step away from the auxiliary lane. For the deployment of our algorithm Bi-AC, we set the car in main lane to be the leader and the car in auxiliary lane to be the follower.
   
    We ran our algorithm Bi-AC together with the baselines for $10$ times using $10$ random seeds to ensure all the algorithms face the same difficulty in training. The result is shown in Table~\ref{table:Traffic} and Fig.~{\ref{fig:highway_merge_rate}}. We found that in our setting, Bi-AC converged to a situation of going first with high probability of $70\%$. Note that the rest $30\%$ probability came to the situation that the auxiliary lane car started with a much higher speed which made it impossible to wait for the leader. We also found that the other baselines failed to choose which car to go first because they do not have a preference so their estimation of the main lane car going first is about $50\%$, as shown in Table~\ref{table:Traffic}. In Fig.~{\ref{fig:highway_merge_rate}} we see that from the very beginning Bi-AC has almost the same possibility of three results: car crush, main lane car going first and auxiliary lane car going first. With the training, Bi-AC makes the possibility of main lane car going first improve steadily until convergence, which shows that Bi-AC can solve real world problems like traffic merge on highway.
    
\section{Conclusion}

    In this paper, we consider Stackelberg equilibrium as a potentially better learning objective than Nash equilibrium in coordination environments due to its certainty and optimality. We formally define the bi-level reinforcement learning problem as the multi-state model-free Stackelberg equilibrium learning problem and empirically study the relationship between the cooperation level and the superiority of Stackelberg equilibrium to Nash equilibrium. We then propose a novel bi-level actor-critic algorithm which is trained centrally and executed decentrally. Our experiments on matrix games and a highway merge environment demonstrate the effectiveness of our algorithm to find the Stackelberg solutions. 

\section{Acknowledgments}
This work is supported by "New Generation of AI 2030" Major Project 2018AAA0100900 and NSFC 61702327, 61632017.

\bibliographystyle{aaai}
\bibliography{bilevel-opt-marl}

\begin{thebibliography}{}

\bibitem[\protect\citeauthoryear{Boutilier}{1996}]{boutilier1996planning}
Boutilier, C.
\newblock 1996.
\newblock Planning, learning and coordination in multiagent decision processes.
\newblock In {\em 6th TARK},  195--210.
\newblock Morgan Kaufmann Publishers Inc.

\bibitem[\protect\citeauthoryear{Bu \bgroup et al\mbox.\egroup
  }{2008}]{bu2008comprehensive}
Bu, L.; Babu, R.; De~Schutter, B.; et~al.
\newblock 2008.
\newblock A comprehensive survey of multiagent reinforcement learning.
\newblock {\em IEEE Transactions on Systems, Man, and Cybernetics, Part C
  (Applications and Reviews)} 38(2):156--172.

\bibitem[\protect\citeauthoryear{Claus and Boutilier}{1998}]{claus1998dynamics}
Claus, C., and Boutilier, C.
\newblock 1998.
\newblock The dynamics of reinforcement learning in cooperative multiagent
  systems.
\newblock {\em AAAI/IAAI} 1998:746--752.

\bibitem[\protect\citeauthoryear{Dempe}{2018}]{dempe2018bilevel}
Dempe, S.
\newblock 2018.
\newblock {\em Bilevel optimization: theory, algorithms and applications}.
\newblock TU Bergakademie Freiberg, Fakult{\"a}t f{\"u}r Mathematik und
  Informatik.

\bibitem[\protect\citeauthoryear{Foerster \bgroup et al\mbox.\egroup
  }{2016}]{foerster2016learning}
Foerster, J.; Assael, I.~A.; de~Freitas, N.; and Whiteson, S.
\newblock 2016.
\newblock Learning to communicate with deep multi-agent reinforcement learning.
\newblock In {\em NIPS},  2137--2145.

\bibitem[\protect\citeauthoryear{Foerster \bgroup et al\mbox.\egroup
  }{2018}]{foerster2018counterfactual}
Foerster, J.~N.; Farquhar, G.; Afouras, T.; Nardelli, N.; and Whiteson, S.
\newblock 2018.
\newblock Counterfactual multi-agent policy gradients.
\newblock In {\em AAAI}.

\bibitem[\protect\citeauthoryear{Fujimoto, van Hoof, and
  Meger}{2018}]{fujimoto2018addressing}
Fujimoto, S.; van Hoof, H.; and Meger, D.
\newblock 2018.
\newblock Addressing function approximation error in actor-critic methods.
\newblock {\em arXiv preprint arXiv:1802.09477}.

\bibitem[\protect\citeauthoryear{Greenwald, Hall, and
  Serrano}{2003}]{greenwald2003correlated}
Greenwald, A.; Hall, K.; and Serrano, R.
\newblock 2003.
\newblock Correlated q-learning.
\newblock In {\em ICML}, volume~3,  242--249.

\bibitem[\protect\citeauthoryear{Haarnoja \bgroup et al\mbox.\egroup
  }{2017}]{haarnoja2017reinforcement}
Haarnoja, T.; Tang, H.; Abbeel, P.; and Levine, S.
\newblock 2017.
\newblock Reinforcement learning with deep energy-based policies.
\newblock In {\em 34th ICML-Volume 70},  1352--1361.
\newblock JMLR. org.

\bibitem[\protect\citeauthoryear{Hu and Wellman}{2003}]{hu2003nash}
Hu, J., and Wellman, M.~P.
\newblock 2003.
\newblock Nash q-learning for general-sum stochastic games.
\newblock {\em JMLR} 4(Nov):1039--1069.

\bibitem[\protect\citeauthoryear{Jang, Gu, and
  Poole}{2016}]{jang2016categorical}
Jang, E.; Gu, S.; and Poole, B.
\newblock 2016.
\newblock Categorical reparameterization with gumbel-softmax.
\newblock {\em arXiv:1611.01144}.

\bibitem[\protect\citeauthoryear{Kapetanakis and
  Kudenko}{2002}]{kapetanakis2002improving}
Kapetanakis, S., and Kudenko, D.
\newblock 2002.
\newblock Improving on the reinforcement learning of coordination in
  cooperative multi-agent systems.
\newblock In {\em 2nd AAMAS}.

\bibitem[\protect\citeauthoryear{K{\"o}n{\"o}nen}{2004}]{kononen2004asymmetric}
K{\"o}n{\"o}nen, V.
\newblock 2004.
\newblock Asymmetric multiagent reinforcement learning.
\newblock {\em Web Intelligence and Agent Systems: An international journal}
  2(2):105--121.

\bibitem[\protect\citeauthoryear{Lauer and
  Riedmiller}{2000}]{lauer2000algorithm}
Lauer, M., and Riedmiller, M.
\newblock 2000.
\newblock An algorithm for distributed reinforcement learning in cooperative
  multi-agent systems.
\newblock In {\em ICML}.
\newblock Citeseer.

\bibitem[\protect\citeauthoryear{Leurent}{2018}]{highway-env}
Leurent, E.
\newblock 2018.
\newblock An environment for autonomous driving decision-making.
\newblock \url{https://github.com/eleurent/highway-env}.

\bibitem[\protect\citeauthoryear{Littman and Stone}{2001}]{littman2001leading}
Littman, M.~L., and Stone, P.
\newblock 2001.
\newblock Leading best-response strategies in repeated games.
\newblock In {\em 7th IJCAI}.
\newblock Citeseer.

\bibitem[\protect\citeauthoryear{Littman}{1994}]{littman1994markov}
Littman, M.~L.
\newblock 1994.
\newblock Markov games as a framework for multi-agent reinforcement learning.
\newblock In {\em Machine learning proceedings 1994}. Elsevier.
\newblock  157--163.

\bibitem[\protect\citeauthoryear{Littman}{2001}]{littman2001friend}
Littman, M.~L.
\newblock 2001.
\newblock Friend-or-foe q-learning in general-sum games.
\newblock In {\em ICML}, volume~1,  322--328.

\bibitem[\protect\citeauthoryear{Lowe \bgroup et al\mbox.\egroup
  }{2017}]{lowe2017multi}
Lowe, R.; Wu, Y.; Tamar, A.; Harb, J.; Abbeel, O.~P.; and Mordatch, I.
\newblock 2017.
\newblock Multi-agent actor-critic for mixed cooperative-competitive
  environments.
\newblock In {\em NIPS},  6379--6390.

\bibitem[\protect\citeauthoryear{Matignon, Laurent, and
  Fort{-}Piat}{2007}]{DBLP:conf/iros/MatignonLF07}
Matignon, L.; Laurent, G.~J.; and Fort{-}Piat, N.~L.
\newblock 2007.
\newblock Hysteretic q-learning : an algorithm for decentralized reinforcement
  learning in cooperative multi-agent teams.
\newblock In {\em 2007 {IEEE/RSJ} ICIRS},  64--69.

\bibitem[\protect\citeauthoryear{Matignon, Laurent, and
  Le~Fort-Piat}{2009}]{matignon2009coordination}
Matignon, L.; Laurent, G.~J.; and Le~Fort-Piat, N.
\newblock 2009.
\newblock Coordination of independent learners in cooperative markov games.

\bibitem[\protect\citeauthoryear{Mintz \bgroup et al\mbox.\egroup
  }{2018}]{mintz2018control}
Mintz, Y.; Cabrera, J.~A.; Pedrasa, J.~R.; and Aswani, A.
\newblock 2018.
\newblock Control synthesis for bilevel linear model predictive control.
\newblock In {\em ACC 2018},  2338--2343.
\newblock IEEE.

\bibitem[\protect\citeauthoryear{Mnih \bgroup et al\mbox.\egroup
  }{2015}]{mnih2015human}
Mnih, V.; Kavukcuoglu, K.; Silver, D.; Rusu, A.~A.; Veness, J.; Bellemare,
  M.~G.; Graves, A.; Riedmiller, M.; Fidjeland, A.~K.; Ostrovski, G.; et~al.
\newblock 2015.
\newblock Human-level control through deep reinforcement learning.
\newblock {\em Nature} 518(7540):529.

\bibitem[\protect\citeauthoryear{Nair and Hinton}{2010}]{nair2010rectified}
Nair, V., and Hinton, G.~E.
\newblock 2010.
\newblock Rectified linear units improve restricted boltzmann machines.
\newblock In {\em ICML-10},  807--814.

\bibitem[\protect\citeauthoryear{Omidshafiei \bgroup et al\mbox.\egroup
  }{2017}]{DBLP:conf/icml/OmidshafieiPAHV17}
Omidshafiei, S.; Pazis, J.; Amato, C.; How, J.~P.; and Vian, J.
\newblock 2017.
\newblock Deep decentralized multi-task multi-agent reinforcement learning
  under partial observability.
\newblock In {\em ICML},  2681--2690.

\bibitem[\protect\citeauthoryear{Palmer \bgroup et al\mbox.\egroup
  }{2018}]{DBLP:conf/atal/PalmerTBS18}
Palmer, G.; Tuyls, K.; Bloembergen, D.; and Savani, R.
\newblock 2018.
\newblock Lenient multi-agent deep reinforcement learning.
\newblock In {\em AAMAS},  443--451.

\bibitem[\protect\citeauthoryear{Peng \bgroup et al\mbox.\egroup
  }{2017}]{peng2017multiagent}
Peng, P.; Wen, Y.; Yang, Y.; Yuan, Q.; Tang, Z.; Long, H.; and Wang, J.
\newblock 2017.
\newblock Multiagent bidirectionally-coordinated nets: Emergence of human-level
  coordination in learning to play starcraft combat games.
\newblock {\em arXiv preprint arXiv:1703.10069}.

\bibitem[\protect\citeauthoryear{Pesce and Montana}{2019}]{pesce2019improving}
Pesce, E., and Montana, G.
\newblock 2019.
\newblock Improving coordination in multi-agent deep reinforcement learning
  through memory-driven communication.
\newblock {\em arXiv preprint arXiv:1901.03887}.

\bibitem[\protect\citeauthoryear{Rabin}{1993}]{rabin1993incorporating}
Rabin, M.
\newblock 1993.
\newblock Incorporating fairness into game theory and economics.
\newblock {\em The American economic review}  1281--1302.

\bibitem[\protect\citeauthoryear{Silver \bgroup et al\mbox.\egroup
  }{2014}]{silver2014deterministic}
Silver, D.; Lever, G.; Heess, N.; Degris, T.; Wierstra, D.; and Riedmiller, M.
\newblock 2014.
\newblock Deterministic policy gradient algorithms.

\bibitem[\protect\citeauthoryear{Sinha, Malo, and
  Deb}{2014}]{sinha2014improved}
Sinha, A.; Malo, P.; and Deb, K.
\newblock 2014.
\newblock An improved bilevel evolutionary algorithm based on quadratic
  approximations.
\newblock In {\em 2014 CEC},  1870--1877.
\newblock IEEE.

\bibitem[\protect\citeauthoryear{Sodomka \bgroup et al\mbox.\egroup
  }{2013}]{sodomka2013coco}
Sodomka, E.; Hilliard, E.; Littman, M.; and Greenwald, A.
\newblock 2013.
\newblock Coco-q: Learning in stochastic games with side payments.
\newblock In {\em ICML},  1471--1479.

\bibitem[\protect\citeauthoryear{Sukhbaatar, Fergus, and
  others}{2016}]{sukhbaatar2016learning}
Sukhbaatar, S.; Fergus, R.; et~al.
\newblock 2016.
\newblock Learning multiagent communication with backpropagation.
\newblock In {\em NIPS},  2244--2252.

\bibitem[\protect\citeauthoryear{Tian \bgroup et al\mbox.\egroup
  }{2019}]{tian2019regularized}
Tian, Z.; Wen, Y.; Gong, Z.; Punakkath, F.; Zou, S.; and Wang, J.
\newblock 2019.
\newblock A regularized opponent model with maximum entropy objective.
\newblock {\em arXiv preprint arXiv:1905.08087}.

\bibitem[\protect\citeauthoryear{Vanderschraaf}{1995}]{vanderschraaf1995convention}
Vanderschraaf, P.
\newblock 1995.
\newblock Convention as correlated equilibrium.
\newblock {\em Erkenntnis} 42(1):65--87.

\bibitem[\protect\citeauthoryear{Von~Stackelberg}{2010}]{von2010market}
Von~Stackelberg, H.
\newblock 2010.
\newblock {\em Market structure and equilibrium}.
\newblock Springer Science \& Business Media.

\bibitem[\protect\citeauthoryear{Wang \bgroup et al\mbox.\egroup
  }{2019}]{wang2019deep}
Wang, Y.; Shi, Z.~R.; Yu, L.; Wu, Y.; Singh, R.; Joppa, L.; and Fang, F.
\newblock 2019.
\newblock Deep reinforcement learning for green security games with real-time
  information.
\newblock In {\em AAAI}, volume~33,  1401--1408.

\bibitem[\protect\citeauthoryear{Wei and Luke}{2016}]{wei2016lenient}
Wei, E., and Luke, S.
\newblock 2016.
\newblock Lenient learning in independent-learner stochastic cooperative games.
\newblock {\em JMLR} 17(1):2914--2955.

\bibitem[\protect\citeauthoryear{Wen \bgroup et al\mbox.\egroup
  }{2019}]{wen2019probabilistic}
Wen, Y.; Yang, Y.; Luo, R.; Wang, J.; and Pan, W.
\newblock 2019.
\newblock Probabilistic recursive reasoning for multi-agent reinforcement
  learning.
\newblock {\em arXiv preprint arXiv:1901.09207}.

\bibitem[\protect\citeauthoryear{Yang \bgroup et al\mbox.\egroup
  }{2018}]{yang2018mean}
Yang, Y.; Luo, R.; Li, M.; Zhou, M.; Zhang, W.; and Wang, J.
\newblock 2018.
\newblock Mean field multi-agent reinforcement learning.
\newblock {\em arXiv preprint arXiv:1802.05438}.

\bibitem[\protect\citeauthoryear{Zhang and Lin}{2012}]{Zhang2012BilevelDS}
Zhang, D., and Lin, G.-H.
\newblock 2012.
\newblock Bilevel direct search method for leader-follower equilibrium problems
  and applications.

\end{thebibliography}

\section{Supplementary Material}

\subsection{Comparison of SE with NE}

    We generate the payoffs for a $n \times n$ matrix game using the multivariate normal distribution with 0 for mean, 1 for variance and various parameters for covariance, which represents the cooperation level of the generated matrix game. When the covariance equals to 1, the game is fully cooperative. $n$ is ranging from $5 \times 5$ to $100 \times 100$. 
    
    \begin{figure}[h]
        \centering
        \includegraphics[width=1.0\columnwidth]{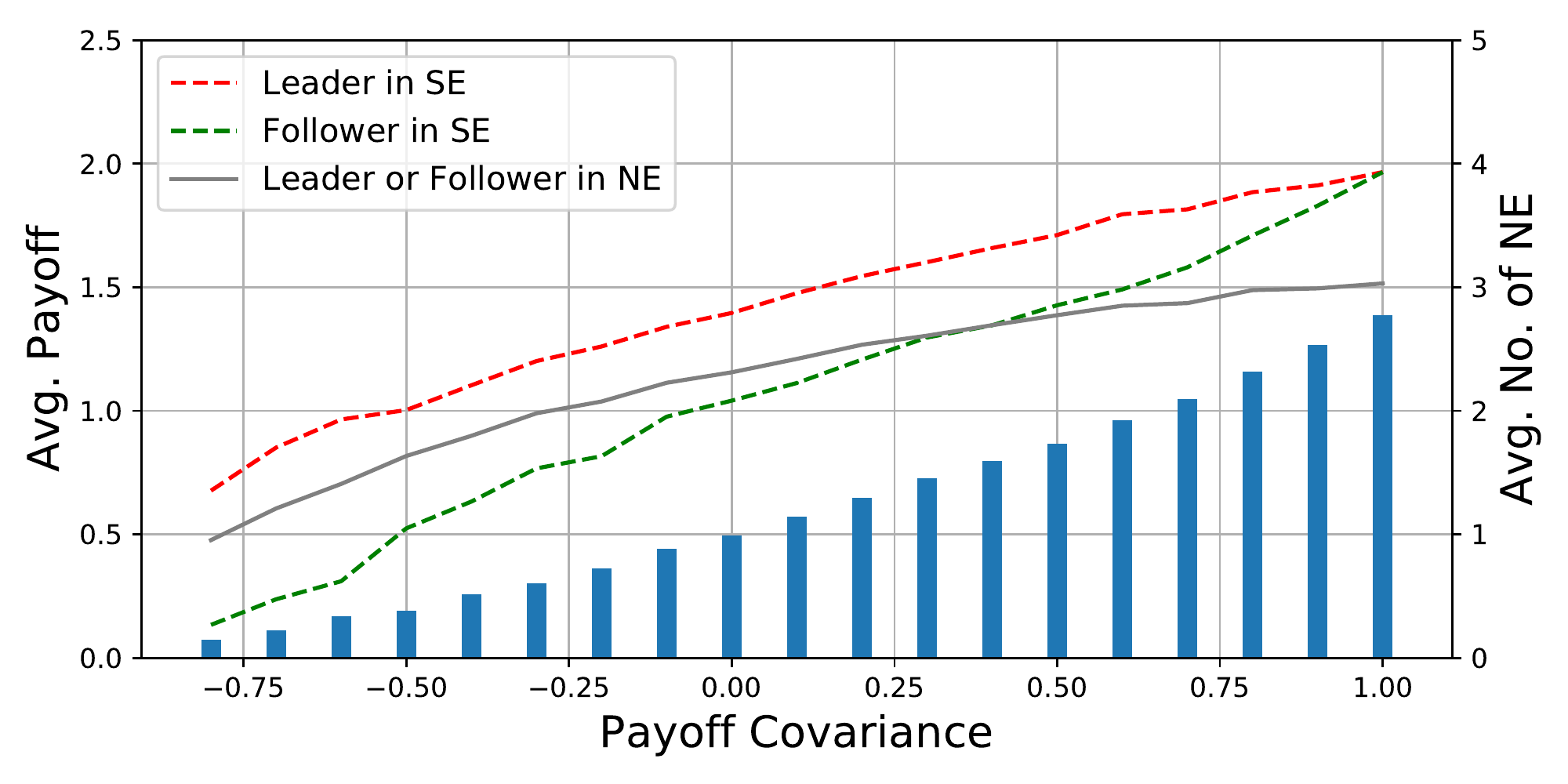}
        \caption{$5 \times 5$ averaged over 2000 trials.}
        \label{fig:bi-ac}
    \end{figure}
    
    \begin{figure}[h]
        \centering
        \includegraphics[width=1.0\columnwidth]{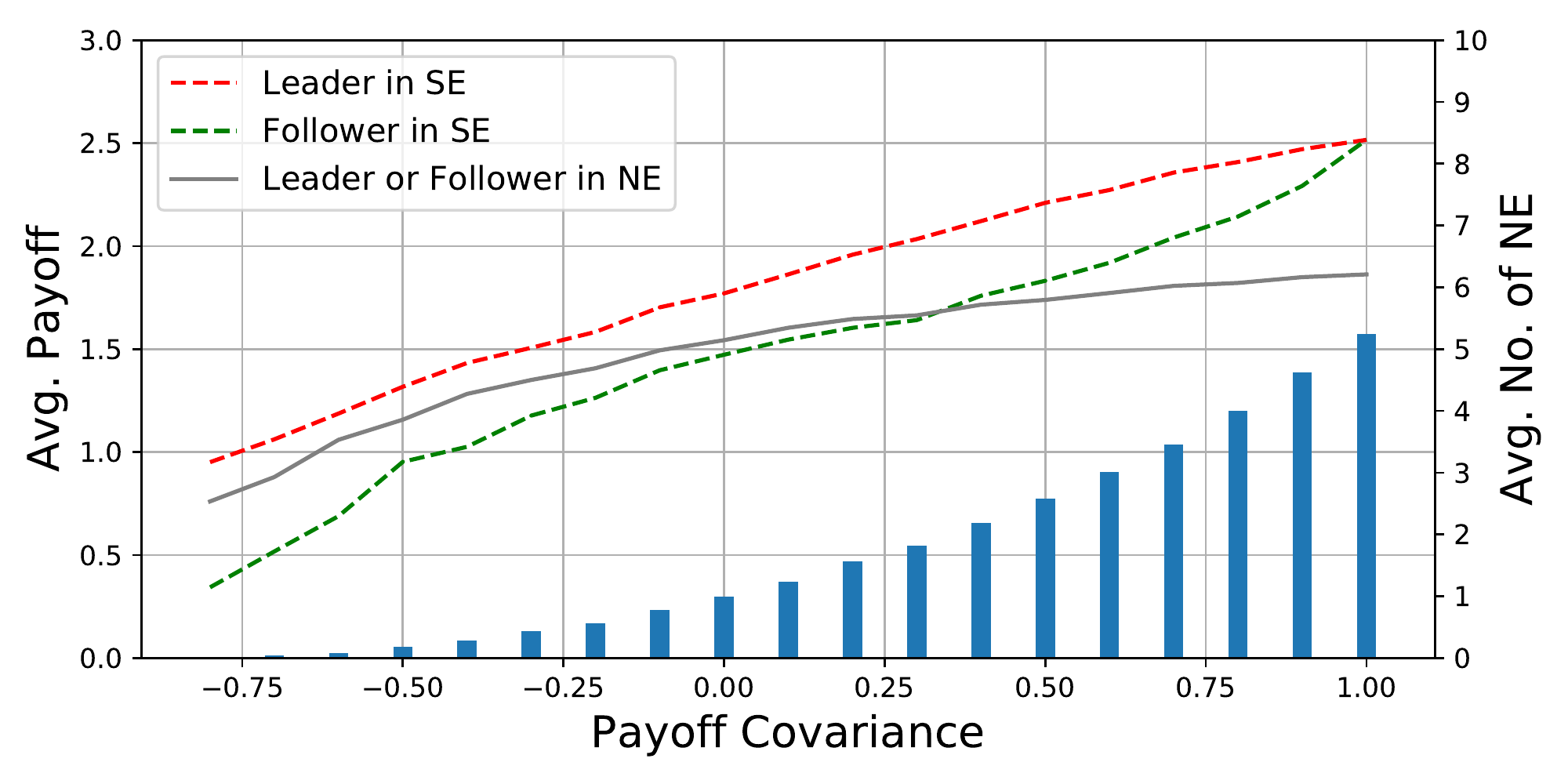}
        \caption{$10 \times 10$ averaged over 2000 trials.}
        \label{fig:bi-ac}
    \end{figure}
    
    \begin{figure}[h]
        \centering
        \includegraphics[width=1.0\columnwidth]{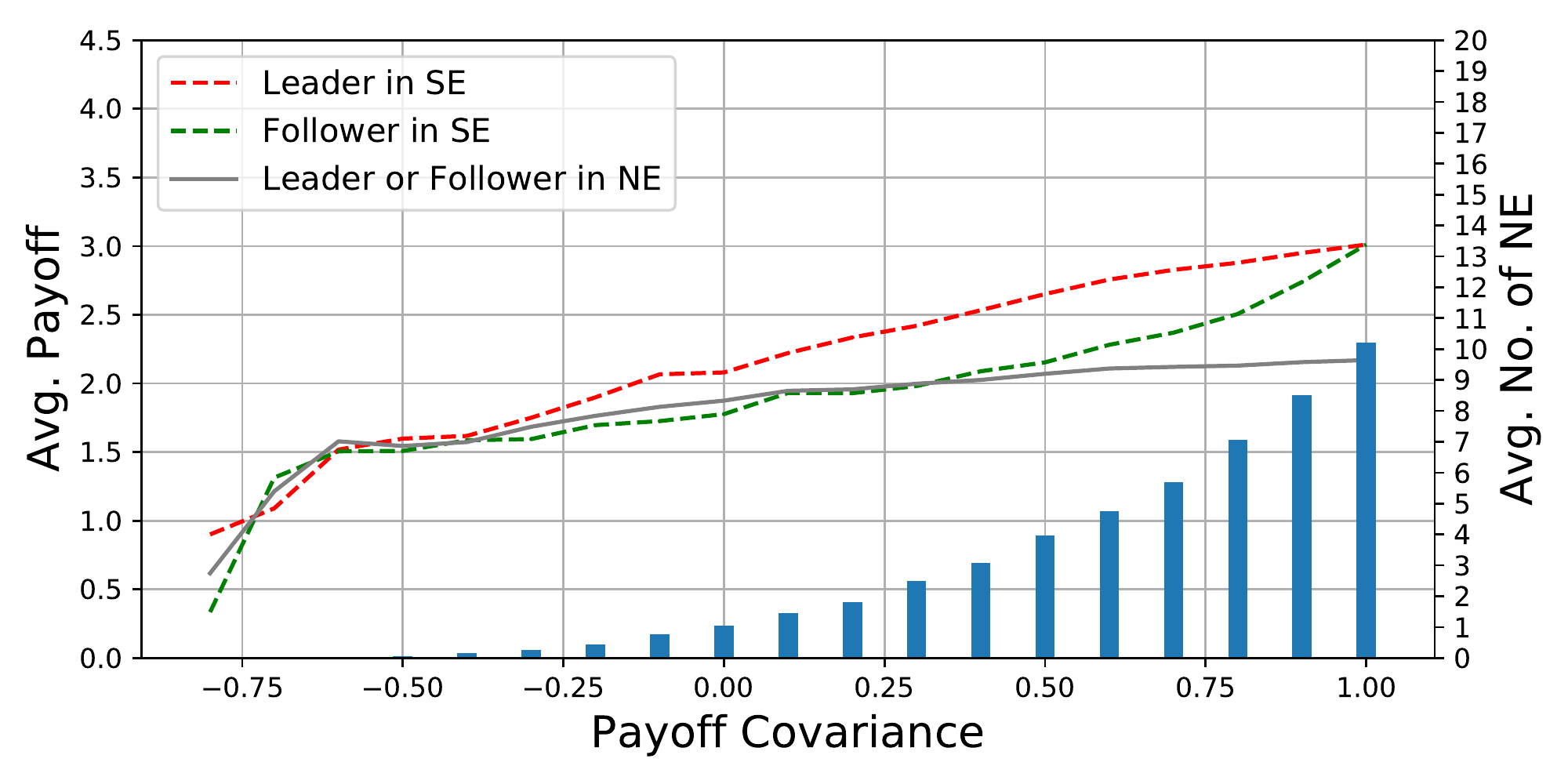}
        \caption{$20 \times 20$ averaged over 500 trials.}
        \label{fig:bi-ac}
    \end{figure}
    
    \begin{figure}[h]
        \centering
        \includegraphics[width=1.0\columnwidth]{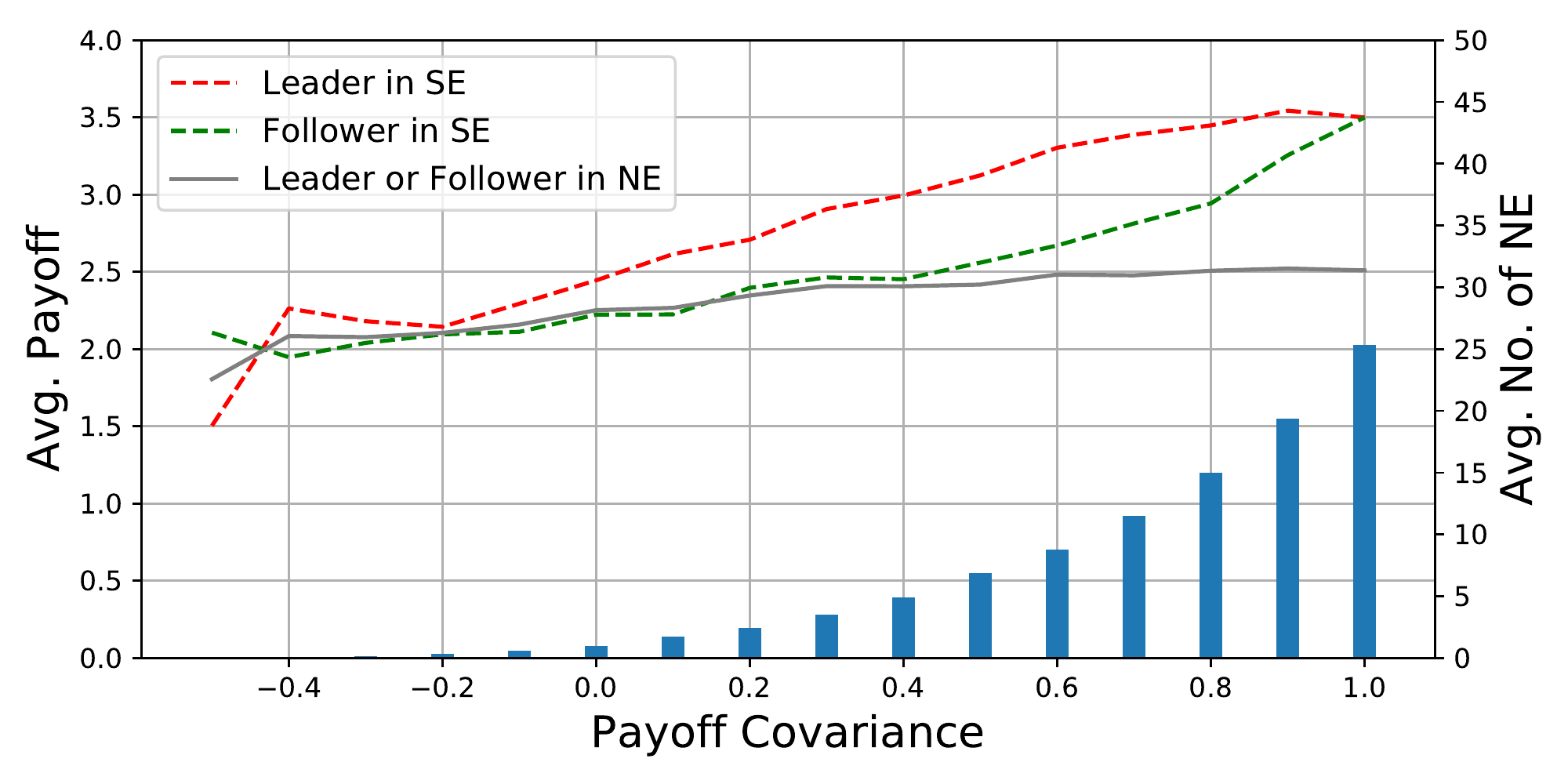}
        \caption{$50 \times 50$ averaged over 200 trials.}
        \label{fig:bi-ac}
    \end{figure}
    
    \begin{figure}[h]
        \centering
        \includegraphics[width=1.0\columnwidth]{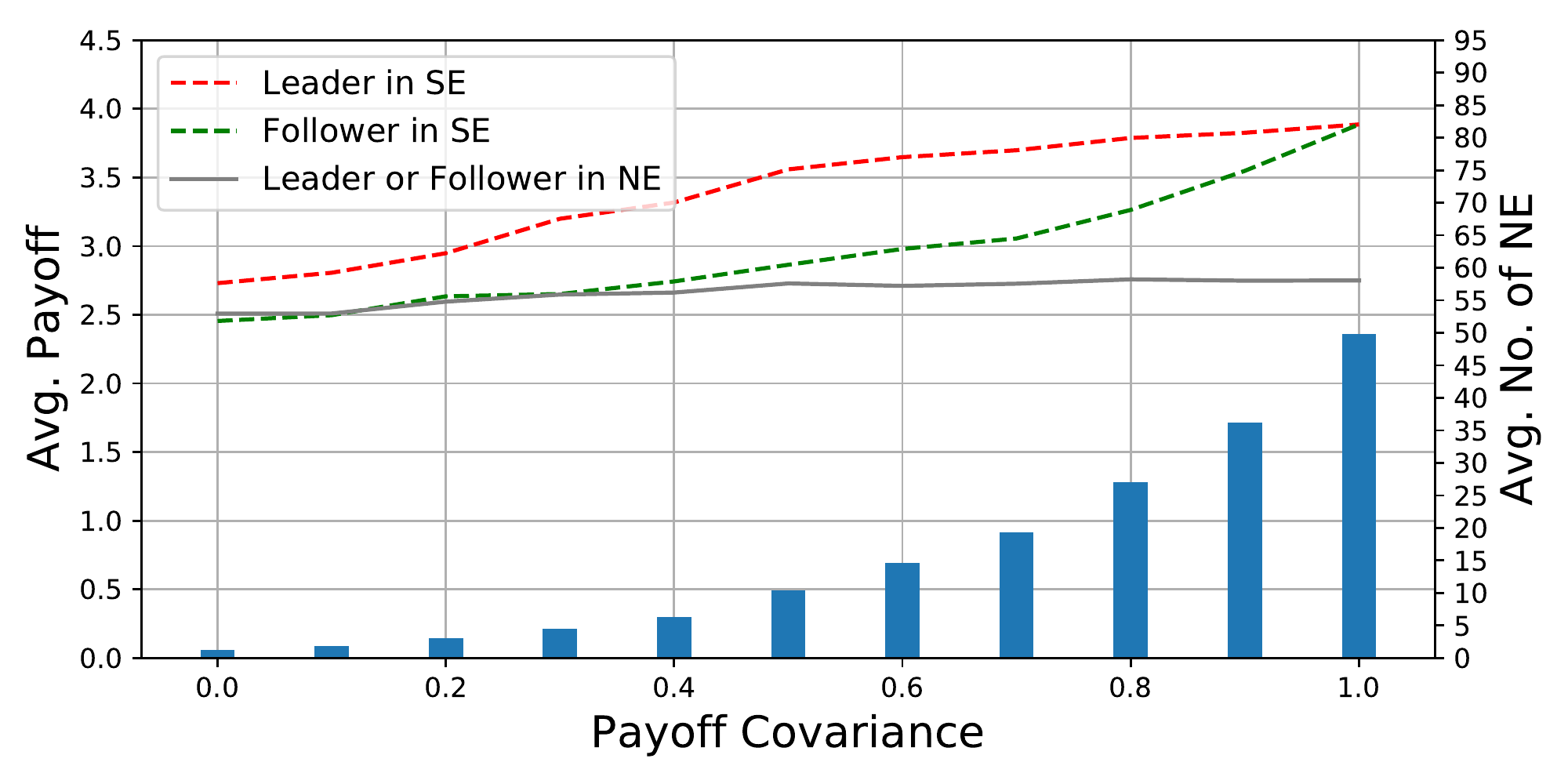}
        \caption{$100 \times 100$ averaged over 100 trials.}
        \label{fig:bi-ac}
    \end{figure}

\newpage

\subsection{Bi-AC Convergence Proof} \label{section:convergence-proof}
    
    Our convergence proof is similar to the proof of Nash-Q in \cite{hu2003nash}.

    \begin{assumption}
        Every state $s \in S$, and action $a_{k} \in A_{k}$ for k = 1, \dots, n, where $A_{k}$ stands for the action space for agent k, are visited infinitely often.
    \end{assumption}

    \begin{assumption}
    \label{a-conv-to-zero}
        The learning rate $\alpha^{t}$ satisfies the following conditions for all $s,t,a_{1},\dots,a_{n}$:
        \begin{enumerate}
            \item $0 \le \alpha^t(s,a_{1}, a_{2})<1$, $\sum_{t=0}^{\infty}\alpha^t(s,a_{1}, a_{2})=\infty, \sum_{t=0}^{\infty}[\alpha^{t}(s,a_{1}, a_{2})]^2 < \infty$, and the latter two hold uniformly and with probability $1$.
            \item $\alpha^t(s,a_{1}, a_{2})=0$ if $(s,a_{1},a_{2}) \ne (s^{t},a^{t}_{1},a^{t}_{2})$
        \end{enumerate}
    \end{assumption}

    \begin{lemma}[Szepesv´ari and Littman (1999), Corollary 5]
    \label{Q-to-zero}
        Assume that $\alpha^t$ satisfies Assumption \ref{a-conv-to-zero} and the mapping $P^t$:$\mathbb{Q}\to \mathbb{Q}$ satisfies the following condition: there exists a number $0 < \gamma < 1$ and a sequence $\lambda^t \ge 0$ converging to zero with probability $1$ such that $||P^tQ-P^tQ^*|| \le \gamma||Q-Q^*|| + \lambda^t$ for all $Q \in \mathbb{Q}$ and $Q^* = E[P^tQ^*]$, then the iteration defined by 
        \begin{align}
            Q^{t+1} = (1 - \alpha)Q^t + \alpha^t[P^tQ^t]
        \end{align}
        converges to $Q^*$ with probability $1$.
    \end{lemma}

    \begin{definition}
        Let $Q = (Q_1, Q_2),$ where $Q_1 \in \mathbb{Q}_{1}, Q_2 \in \mathbb{Q}_{2}$, and $\mathbb{Q} = \mathbb{Q}_1 \times \mathbb{Q}_2$. $P^t: \mathbb{Q} \to \mathbb{Q}$ is a mapping on the complete metric space $\mathbb{Q} \to \mathbb{Q}$, $P^tQ = (P^tQ_1, P^tQ_2)$, where 
        \begin{align*}
            P^tQ_k(s, a_1, a_2) = r^t_k(s, a_1, a_2) + 
            \gamma Q_k(s', \pi_1(s'), \pi_2(s')), \\ \text{for k = 1,2}
        \end{align*}
        where $s'$ is the state at time $t+1$.
    \end{definition}




    \begin{lemma}
    \label{q_equil_point}
        For any $s \in S, a_1 \in A_2, a_2 \in A_2$, a convergence point $Q^*$ in the stage game satisfies
        \begin{align*}
            &Q^*_k(s,a_1, a_2) = r_k(s,a_1,a_2) + \\
            &\gamma\sum_{s'\in S}p(s'|s, a_1, a_2)Q^*_k(s', \pi_1(s'), \pi_2(s')), \text{for k = 1,2}
        \end{align*}
        
    \end{lemma}

    If the equation is not satisfied, in the next step of update, the Q-value will change and thus this is not a equilibrium point.

    \begin{lemma}
    \label{expected value}
        For a two player stochastic game, $E[P^tQ^*]=Q^*$, where $Q^*=(Q^*_1, Q^*_2)$.
    \end{lemma}
    
    \begin{proof}
        We use $\theta$ to represent the distribution of the state transition given current state $s$, action of the leader $a_1$, and the action of the follower $a_2$.
        
        Using Lemma \ref{q_equil_point}, we can get the following equations:
        \begin{align}
            &Q^*_k(s,a_1,a_2) \\
            &= r_k(s,a_1,a_2) + 
            \gamma\sum_{s'\in S}p(s'|s, a_1, a_2)Q^*_k(s', \pi_1(s'), \pi_2(s'))\\
            &= \sum_{s' \in S}p(s'|s,a_1,a_2)(\gamma Q^*_k(s', \pi_1(s'), \pi_2(s')) +  r_k(s,a_1,a_2))\\
            & = E_\theta[P^t_kQ^*_k(s,a_1,a_2)]
        \end{align}
        i.e., $E[P^tQ^*]=Q^*$.
    \end{proof}
    
    \begin{definition}
    \label{global-optimal}
    A joint action $a^* = \langle a_1, a_2 \rangle$ of the stage game in state $s$ is a global optimal point if every agent receive its highest payoff at this point. That is, for all k, 
    \begin{align}
        Q_k(s, a^*) \ge Q_k(s, a), \forall a \in A
    \end{align}
    \end{definition}

    \begin{assumption}
        Every stage game $(Q^t_1(s), Q^t_2(s))$ for all t and s, has a global optimal point , and agents' payoffs in this point are selected by the actor function to update the critic functions with probability $1$.
    \end{assumption}

     Note that in a two-agent game, if we have a global optimal point, we will find it as long as we get the correct actor function, . The assumption is exactly the same with Condition A in Assumption 3 of \cite{hu2003nash}, thus we can use the following definition to help proving the convergence.
    
    \begin{definition}[\cite{hu2003nash}, Definition 15]
        \begin{align}
        &||Q - \hat{Q}|| \\
        &\equiv max_{j,s}||Q^{j}(s) - \hat{Q}^{j}(s)||_{(j,s)} \\
        &\equiv max_{j,s,a_1,a_2}|Q^{j}(s,a_1, a_2) - \hat{Q}^{j}(s,a_1,a_2)|
        \end{align}
    \end{definition}
    
    \begin{lemma}[\cite{hu2003nash}, Lemma 16]
    \label{contract}
        $||P^tQ - P^t\hat{Q}|| \le \gamma ||Q-\hat{Q}||, \forall Q, \hat{Q} \in \mathbb{Q}$.
    \end{lemma}
    
    \begin{theorem}
        Under assumption 1 - 3, the sequence $Q_{t} = (Q^t_1, Q^t_2)$, updated by 
        \begin{align}
            &\pi_1' \leftarrow \argmax_{\pi_1} Q_1(s', \pi_1, \argmax_{a_2} Q_2(s', \pi_1, a_2)), \\ 
            &a_2' \leftarrow \argmax_{a_2} Q_2(s', \pi_1', a_2), \\ 
            &Q_1(s, a_1, a_2) \leftarrow (1-\alpha)Q_1(s, a_1, a_2) \nonumber\\ 
            &~~~~~~~~~~~~~~~~~~~~~~~~~~+ \alpha (r_1 + \gamma Q_1(s', \pi_1', a_2')), \\
            &Q_2(s, a_1, a_2) \leftarrow (1-\alpha)Q_2(s, a_1, a_2)  \nonumber\\ 
            &~~~~~~~~~~~~~~~~~~~~~~~~~~+ \alpha (r_2 + \gamma Q_2(s', \pi_1', a_2')).
        \end{align}
        converges to a fixed value $Q^* = (Q^*_1, Q^*_2)$.
    \end{theorem}
    
    \begin{proof}
        First, $P^t$ is a contraction operator by Lemma \ref{contract}. Second, the fixed point condition, $E(P^tQ^*) = Q^*$ is established by Lemma \ref{expected value}. Then from Lemma \ref{Q-to-zero}, we know the Q value will converge to $Q^*$ with probability 1.
    \end{proof}

    Therefore, we know that the critic function is trained properly, and the actor function of followers can be corrected after the critic converge. 

\subsection{Discussion of convergence of Bi-AC}

    The former convergence proof is a proof  in a sufficient but obviously  not necessary condition. The algorithm may even not satisfy the condition in the very beginning of the training process given that a global optimal rarely appear in multi-agent setting. Even in the matrix game Maintain (table \ref{table:Matrix_game_2}) we showed in our paper, the global optimal point does not exists. 
    
    Nevertheless, in our experiment as reported in the paper, we found that the algorithm was easy to converge. This may because of that if the leading agent's action is already fixed in all possible following state, the training process of the follower agent, no matter it is the leader or the follower, can be transformed into a single agent scenario. Thus, if trained with infinite time and no state loop occurs, our algorithm can always converge given the idea of backward induction in extensive form game. This property is different from the traditional MARL algorithms in that other algorithm require the policy of other agents to be fixed in the current state, rather than the state afterward. 
    
    And as our experiments show, even with the allowance of revisiting the states, our algorithm still converge quickly. This suggests that there may be some potential to relax the conditions in our convergence proof or there may even be another total different method to prove the convergence of our algorithm, at least for some classes of games.

\subsection{An example uncovered by Bi-AC}

    We are able to construct a game and a corresponding joint policy such that the joint policy induces Stackelberg equilibrium in each state but cannot induce Stackelberg equilibrium for the entire game. 

    We consider a game with two agents ($X$ and $Y$) and two states ($s_1$ and $s_2$). Each agent has two actions ($A$ and $B$). $X$ is the leading agent and $Y$ is the following agent. We use $R$ to denote the reward function for agents where $(1,2)$ means $X$ gets reward $1$ and $Y$ gets reward $2$. Also, we use $T$ to denote the transition function. Specifically, let 
    \begin{align}
        &T(s_1, A, A) = s_1, \nonumber \\
        &T(s_1, A, B) = s_2, \nonumber   \\
        &T(s_1, B, A) = s_2, \nonumber \\
        &T(s_1, B, B) = s_1, \nonumber \\
        &R(s_1, A, A) = (0, 0), \nonumber \\
        &R(s_1, A, B) = (0, 10),\nonumber \\
        &R(s_1, B, A) = (10, 0), \nonumber \\
        &R(s_1, B, B) = (-1,-1),  \nonumber 
    \end{align}
    where $s_2$ is a terminal state where reward function and transition function are inapplicable. Intuitively, the BiRL solution is the joint policy $(B, A)$ in state $s_1$, which induces payoff $(10, 0)$. However, we find that $(A, B)$ also satisfies bi-level Bellman equation, although it is not a solution for BiRL. We write down the Q table w.r.t. the joint policy $(A, B)$:
    \begin{align}
        &Q_X(s_1, B, A) = 10, \nonumber \\
        &Q_X(s_1, A, B) = 0, \nonumber \\
        &Q_X(s_1, B, B) = -1, \nonumber \\
        &Q_X(s_1, A, A) = 0,\nonumber \\
        &Q_Y(s_1, B, A) = 0, \nonumber \\
        &Q_Y(s_1, A, B) = 10, \nonumber \\
        &Q_Y(s_1, B, B) = 9, \nonumber \\
        &Q_Y(s_1, A, A) = 10. \nonumber
    \end{align}
    Correspondingly, the value functions are: 
    \begin{align}
        V_X(s_1) = 0,
        V_Y(s_1) = 10.
    \end{align}
    We can verify that the above Q table satisfies the Bellman optimality equation, i.e. is a possible convergent point. Specifically, when $X$ chooses $B$ in $s_1$, $Y$ will also choose $B$ according to $Q_Y$ table, which results in payoff $(-1, 9)$. Otherwise, when $X$ chooses $A$, $Y$ will choose $B$, which results in payoff $(0, 10)$. Since $-1 < 0$, the Stackelberg equilibrium in $s_1$ is $(A, B)$.                                                                             
    
    In general, this situation is caused by the cyclic transition and the TD update of Bellman equation. Since $T(s_1, B, B) = s_1$, the value of $Q_Y(s_1, B, B)$ depends on $V_Y(s_1)$, where $s_1$ is the next state and $(A, B)$ is the policy supposed to follow. Although we explore $(B, B)$ in $s_1$, the policy applying to the subsequent state $s_1$ is still $(A, B)$, which results in a local stable point $(A, B)$. 
    
\end{document}